\documentclass[journal]{new-aiaa}
\usepackage[utf8]{inputenc}
\usepackage{textcomp}

\usepackage{amssymb,amsmath,latexsym,amsthm,amsfonts,mathtools}

\usepackage{mathrsfs}
\usepackage{multirow}
\usepackage{graphicx}
\usepackage{textcomp,siunitx}
\usepackage{float}
\usepackage{subcaption}
\usepackage{booktabs}
\usepackage{enumitem}
\usepackage{tikz}
\usepackage{array}
\usepackage{siunitx}
\usepackage{longtable,tabularx}
\usepackage{makecell}
\usepackage{booktabs}
\theoremstyle{plain}
\newtheorem{theorem}{Theorem}
\newtheorem{lemma}{Lemma}

\usepackage{amsmath}
\usepackage{bbding}
\theoremstyle{remark}
\newtheorem{remark}{Remark}

\usepackage{epstopdf} 
\theoremstyle{definition}
\newtheorem{definition}{Definition}

\newtheorem*{problem*}{Problem}
\usetikzlibrary{calc}
\usetikzlibrary{arrows.meta, calc, decorations.pathreplacing}

\newcommand{\sign}{\mathrm{sign}}

\usepackage{cleveref}
\usepackage{hyperref}
\Crefformat{figure}{#2Fig.~#1#3}
\Crefmultiformat{figure}{Figs.~#2#1#3}{ and~#2#1#3}{, #2#1#3}{ and~#2#1#3}
\Crefformat{equation}{#2Eq.~#1#3}
\Crefmultiformat{equation}{#2Eqs.~#1#3}{ and~#2#1#3}{, #2#1#3}{ and~#2#1#3}
\Crefname{algocfline}{Algorithm}{Algorithms}

\crefname{assumption}{assumption}{assumptions}
\Crefname{assumption}{Assumption}{Assumptions}

\usepackage[linesnumbered,ruled]{algorithm2e}
\SetKwInput{KwInput}{Input}         
\SetKwInput{KwOutput}{Output}

\usepackage{newtxmath}

\usepackage{anyfontsize}
\usepackage{accents}

\title{Terminal Time and Angle-Constrained Nonlinear Intercept Guidance}

\author{Shivam Bajpai\footnote{Ph.D. Research Scholar, \textbf{email}: \texttt{bajpaism@mail.uc.edu}}, and Abhinav Sinha \footnote{Assistant Professor, \textbf{email}: \texttt{abhinav.sinha@uc.edu} (Senior Member, AIAA).}}
\affil{Guidance, Autonomy, Learning, and Control for Intelligent Systems (GALACxIS) Lab,\\ Department of Aerospace Engineering and Engineering Mechanics,\\ University of Cincinnati, Cincinnati, OH, 45221,  USA}
\begin{document}
\maketitle
\begin{abstract}
This paper considers the problem of simultaneously controlling an interceptor's impact time and impact angle using its lateral acceleration as the sole control input. With a single control input, the nonlinear engagement kinematics is inherently underactuated, which complicates guidance law synthesis. To overcome this challenge, a hierarchical sliding mode-based guidance law is developed to concurrently regulate the two terminal constraints. The proposed architecture consists of a two-layer sliding manifold. The first layer comprises two sub-sliding surfaces corresponding to the impact time and impact angle error dynamics, respectively, while the second layer introduces a composite sliding manifold that combines the two individual sub-surfaces. Then, a variable-gain adaptive guidance law is designed to ensure time and angle-constrained interception against a stationary target, which is further extended to intercept a constant velocity target. Simulations are conducted for various engagement scenarios to attest to the efficacy of the proposed approach.
\end{abstract}
\section{Introduction}\label{sec:intro}
\lettrine[]{I}{n} last few years, rapid advancements in autonomous systems have necessitated the demand for terminal performance by requiring control over impact time and impact angle simultaneously to enhance mission effectiveness in adversarial scenarios. Control over the impact angle enables the interceptor to realize a prescribed terminal approach geometry, which is often essential for exploiting directional vulnerabilities or satisfying mission-dependent engagement requirements \cite{4103230}. At the same time, control over the impact time is central to coordinated multi-interceptor operations, especially in salvo settings where temporal synchronization can significantly improve the likelihood of mission success against defended targets \cite{1597196}. Despite substantial progress on impact-time guidance and impact-angle guidance as largely separate problems, the development of a unified guidance framework that achieves interception under both terminal constraints using only the interceptor's lateral acceleration remains a challenging and still unresolved problem.

A substantial body of literature has addressed the design of guidance laws for enforcing a prescribed impact angle at interception \cite{4103230,doi:10.2514/3.57212,640285,doi:10.2514/1.58454,doi:10.2514/1.G001547,doi:10.2514/1.37864,doi:10.2514/1.62910,Lee2012,Hou2017,doi:10.2514/1.60133,doi:10.2514/1.62737}. Much of this development has been rooted in proportional-navigation (PN)-type constructions \cite{4103230,doi:10.2514/3.57212,640285,doi:10.2514/1.58454,doi:10.2514/1.G001547,doi:10.2514/1.37864}. The earliest studies established the feasibility of shaping terminal approach geometry through modified PN laws, first for reentry guidance \cite{4103230} and subsequently in the presence of first-order autopilot dynamics \cite{doi:10.2514/3.57212}. Later works introduced biased PN formulations to explicitly drive the impact-angle error to zero \cite{640285,doi:10.2514/1.58454}, while other variants combined PN with pursuit-based strategies or switching logic to improve terminal angle regulation against stationary or nonmaneuvering targets \cite{doi:10.2514/1.G001547,doi:10.2514/1.37864}. Beyond the PN framework, optimal-control-based formulations have also been developed to achieve impact-angle regulation from a more systematic synthesis perspective \cite{doi:10.2514/1.62910,Lee2012}. These include linearized optimal guidance laws for moving-target interception \cite{doi:10.2514/1.62910} and linear time-varying designs derived through inverse optimality arguments to jointly reduce terminal miss and angle error \cite{Lee2012}. In parallel, sliding-mode-based strategies have been proposed to enhance robustness with respect to target maneuvers, model uncertainty, and bounded disturbances \cite{Hou2017,doi:10.2514/1.60133,doi:10.2514/1.62737}. Representative developments include dual-surface sliding-mode guidance with model-free target acceleration estimation for maneuvering targets \cite{Hou2017}, three-dimensional partial integrated guidance laws with angle constraints \cite{doi:10.2514/1.60133}, and finite-time sliding-mode designs that guarantee convergence to the desired terminal geometry \cite{doi:10.2514/1.62737}.

Against this backdrop, significant attention has also been devoted to guidance strategies that enforce impact time constraints, motivated by the need for simultaneous target interception and enhanced mission effectiveness \cite{1597196,doi:10.2514/1.G001618,doi:10.2514/1.G007174,doi:10.2514/1.G004669,8395020,doi:10.2514/1.G005180,doi:10.2514/1.G004424,doi:10.2514/1.G006190,doi:10.2514/1.G007595,doi:10.2514/1.G008834}. The earliest formulation of this problem appears in \cite{1597196}, where a guidance strategy was developed by augmenting the PN guidance with the feedback of impact time error. Subsequent efforts extended this line of research in several directions. A modified pure PN guidance law was proposed in \cite{doi:10.2514/1.G001618} to achieve interception of stationary targets at a prescribed impact time under nonlinear engagement kinematics, while the work in \cite{doi:10.2514/1.G007174} introduced a varying-gain PN guidance strategy capable of precise impact-time regulation without resorting to linearized engagement models. Beyond shaping the navigation constant, alternative design paradigms have also been explored. In \cite{doi:10.2514/1.G004669}, a data-driven impact-time guidance law was developed to reduce timing error, with particular relevance to hypersonic engagements involving substantial velocity variation. Sliding-mode-based formulations have likewise been employed, including the construction in \cite{8395020}, where a sliding surface combining relative range and time-to-go was used to guarantee interception at a prescribed time, and the three-dimensional nonlinear design in \cite{doi:10.2514/1.G005180}, which addressed impact-time-constrained interception within an SMC framework. In addition, the work in \cite{doi:10.2514/1.G005367} developed a salvo guidance strategy against a stationary target. The authors in \cite{doi:10.2514/1.G006190} derived a guidance strategy by solving a minimum-effort optimal control problem with fixed impact time and a quadratic approximation of the kinematic equations. Based on the relative virtual framework and the classical differential geometry curve theory, a polynomial guidance method was proposed in \cite{doi:10.2514/1.G007595} to intercept maneuvering targets at the desired impact time. In \cite{doi:10.2514/1.G008834}, an output trajectory shaping guidance law with a Bezier curve was proposed for impact time control against a stationary target.

Unlike the formulations that regulate only one terminal objective, the simultaneous treatment of both impact time and angle substantially increases the difficulty of the guidance design. Recently, the problem of simultaneously controlling impact time and impact angle has attracted significant attention. Early developments along this direction include radial-tangential guidance formulations in \cite{9266104,9348900}. The work in \cite{9266104} developed multi-interceptor guidance laws based on finite-time sliding-mode control and the super-twisting algorithm, whereas the authors in \cite{9348900} proposed a cooperative strategy founded on fixed-time control and a leader-follower architecture to achieve simultaneous arrival under prescribed impact-angle constraints. In parallel, several studies pursued time-to-go-estimation-independent designs \cite{doi:10.2514/1.G003765,doi:10.2514/1.G000414,doi:10.2514/1.G006229,10225257}. These include the shaping-function-based analytical construction in \cite{doi:10.2514/1.G003765}, the optimal-control formulation in \cite{doi:10.2514/1.G000414}, the time-varying lead-angle tracking law in \cite{doi:10.2514/1.G006229}, and the geometric guidance approach reported in \cite{10225257}. Sliding-mode-based solutions have also continued to play a prominent role. The work in \cite{Chen_Wang} addressed the simultaneous enforcement of impact time and impact angle through a dedicated sliding-mode structure, while \cite{doi:10.1177/09544100211029817} incorporated explicit time-to-go estimation within the sliding mode framework. Related developments include the polynomial time-to-go based design in \cite{6060928} and the two-stage strategy in \cite{Hu2018NewIT}, which combines sliding mode and PN guidance through a virtual-target switching mechanism. In our previous work \cite{Scitech26}, this line of research was further extended to account for first-order autopilot dynamics, and the interception of a non-maneuvering target via the concept of predicted interception point \cite{doi:10.2514/1.G007122}.

Although the aforementioned studies have significantly advanced the state of the art, several limitations remain. Note that a class of previous methods is based on 
radial-tangential guidance formulations for regulating the desired impact time and impact angle. While such formulations are analytically convenient, their practical realizability is limited in many interceptor settings. The reason is that direct manipulation of the radial velocity component is generally not available as an independent control input. Instead, the interceptor is typically equipped to generate only lateral acceleration normal to its velocity vector, while its longitudinal speed is either fixed, weakly actuated, or constrained by propulsion and airframe limitations. Consequently, guidance laws that presuppose direct regulation of radial motion may not be consistent with the true control authority of the vehicle and may therefore be difficult to implement in realistic engagements. A different class of methods achieves simultaneous terminal constraint satisfaction by imposing a particular sliding-surface structure or by explicitly embedding a time-to-go estimate into the guidance design. Although such constructions can be effective under nominal conditions, they often inherit reduced flexibility and robustness. Hence, the guidance performance gets largely dictated by the chosen surface parameterization or by the fidelity of the chosen time-to-go estimate. Moreover, optimal-control-based and polynomial-parameterization-based schemes often incur a nontrivial computational burden, especially when extended to cooperative or multi-interceptor scenarios. In such settings, the simultaneous enforcement of impact time and angle constraints across multiple agents can lead to increased online computational complexity, which may limit scalability and real-time deployability. Accordingly, there remains a clear need for guidance designs that enforce both terminal constraints through the interceptor's available lateral-acceleration input alone, without being tied to a particular sliding-mode structure or to explicit time-to-go estimation. We summarize the main contribution of this work below.

A general hierarchical sliding-manifold structure is introduced to regulate the terminal objectives through submanifolds. This construction provides a systematic mechanism for decomposing the simultaneous impact-time and impact-angle regulation problem into interconnected objectives while preserving an integrated design architecture. Owing to this general construction, the proposed method is flexible, not tied to any particular time-to-go expression, and can accommodate arbitrary time-to-go estimates and varying target motions. 

A variable-gain adaptive guidance law is proposed for the simultaneous control of impact time and impact angle using the interceptor's lateral acceleration as the sole control input. This is particularly important from a practical standpoint, since lateral acceleration constitutes the physically available guidance input in many interceptor systems, whereas direct regulation of radial motion is generally not realizable. By constructing the guidance law directly through the available control channel, the proposed design remains consistent with realistic interceptor control authority and is therefore more suitable for practical implementation.

The resulting framework is modular and flexible, making it amenable to the incorporation of additional terminal or in-flight constraints without fundamentally altering the baseline guidance architecture. Any new constraint can be accommodated by introducing an associated submanifold, after which an augmented composite manifold may be constructed to integrate the new objective within the existing framework.

The proposed guidance framework is formulated directly under nonlinear engagement kinematics, which avoids reliance on linearized approximations and broadens its applicability across diverse operating conditions. The design is first established for stationary targets and is then extended to non-maneuvering targets through the predicted interception point concept, thus enlarging the class of target scenarios that can be handled within a unified framework.

The proposed guidance law is systematically compared with existing (simultaneous) impact time and impact angle guidance strategies to highlight its structural and practical advantages. In addition, the efficacy and generality of the proposed framework are further demonstrated by redesigning the guidance law using a different time-to-go expression, thereby showing that the methodology is not tied to a particular time-to-go model and can retain its effectiveness under alternative time-to-go constructions.

\section{Problem Formulation}\label{sec:problem}
As shown in \Cref{fig:enggeo}, we consider a planar two-agent engagement scenario consisting of an interceptor ($\mathrm{P}$) and a target ($\mathrm{T}$).
 \begin{figure}[h!]
    \centering
\includegraphics[width=.75\linewidth]{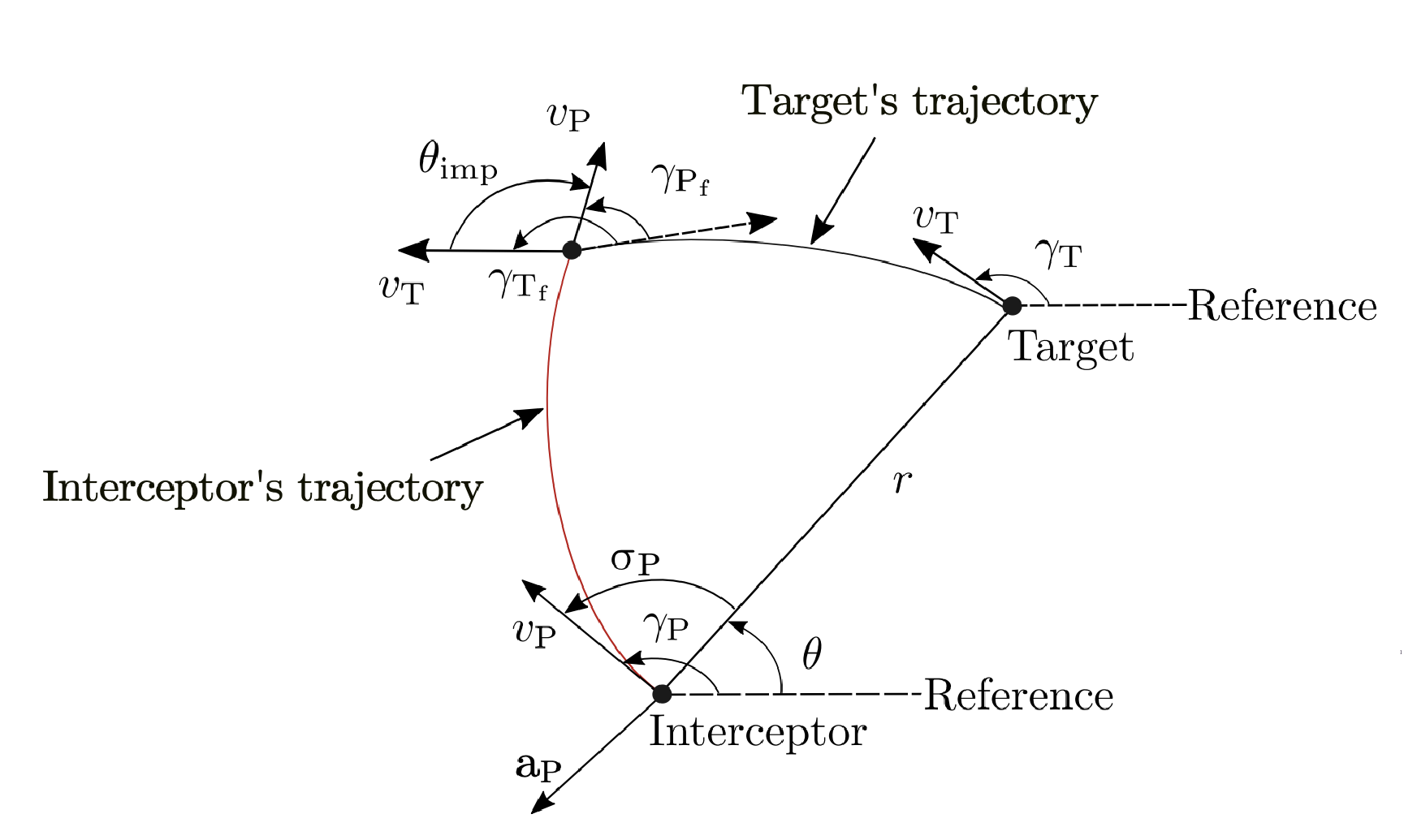}
    \caption{Interceptor-target engagement geometry in a plane.} 
    \label{fig:enggeo}
\end{figure}
The interceptor is assumed to move with constant speed $v_\mathrm{P}$ and is tasked with intercepting the target subject to prescribed terminal constraints on both the impact time and the impact angle. The planar engagement assumption is appropriate for scenarios in which the interaction occurs at approximately constant altitude and captures the essential geometric features of many air-combat settings. In \Cref{fig:enggeo}, $r$ and $\theta$ denote the relative range and the line-of-sight (LOS) angle, respectively, while $\gamma_\mathrm{P}$ and $\sigma_\mathrm{P}$ represent the interceptor heading angle and lead angle. Under these definitions, the kinematics of relative motion between the agents are given by
\begin{subequations}\label{eq:enggeo}
    \begin{align} 
        \dot{r}=&-v_\mathrm{P}\cos\sigma_\mathrm{P} \label{eq: r_dot}, \\
        r\dot{\theta}=&-v_\mathrm{P}\sin\sigma_\mathrm{P}  \label{eq: theta_dot},
    \end{align}
    where $\sigma_\mathrm{P}=\gamma_\mathrm{P}-\theta$, and 
    \end{subequations}
    \begin{align}
    \dot{\gamma}_\mathrm{P}=&~\dfrac{a_\mathrm{P}}{v_\mathrm{P}} , \label{eq:gamma_dot}
\end{align}
where $a_\mathrm{P}$ is the interceptor's sole control authority.
\begin{lemma} \label{lemma 2}
    The dynamics of the lead angle of the interceptor has a relative degree of one with respect to its lateral acceleration, $a_\mathrm{P}$.
    \end{lemma}
\begin{proof}
    Differentiating the lead angle $\sigma_\mathrm{P}$ with respect to time yields
  \begin{equation} 
        \dot{\sigma}_\mathrm{P}  = \dot{\gamma}_\mathrm{P} - \dot{\theta}, \label{eq: sigma_dot} 
  \end{equation}
  which can be further simplified using \eqref{eq: theta_dot} and \eqref{eq:gamma_dot} to
  \begin{equation}
      \dot{\sigma}_\mathrm{P}  = \dfrac{a_\mathrm{P}}{v_\mathrm{P}} + \dfrac{v_\mathrm{P} \sin \sigma_\mathrm{P}}{r}. \label{eq: sigma_dot_expression}
  \end{equation}
   It can be observed that from \eqref{eq: sigma_dot_expression} that the interceptor's lateral acceleration appears in the first derivative of the lead angle, indicating that it influences ${\sigma}_\mathrm{P}$. This completes the proof.
\end{proof}
\Cref{lemma 2} indicates that the lead angle responds directly to the interceptor's lateral acceleration. Therefore, by modulating the lateral acceleration, one shapes the evolution of the lead angle, which in turn modifies the interceptor trajectory, the engagement geometry, and ultimately the terminal conditions of the engagement.
\begin{definition}\label{def:time}
    Impact time, $t_\mathrm{f}$, is the time instant when the interceptor intercepts the target. At any instant of time, the remaining time till interception (also known as the time-to-go), $t_\mathrm{go}$, can be defined as the difference between $t_\mathrm{f}$ and the current time, $t$, that is $t_\mathrm{go} = t_\mathrm{f} - t$.
\end{definition}
\begin{definition}\label{def:angle}
    Impact angle, $\theta_\mathrm{imp}$, is the angle between the velocity vectors of the target and the interceptor at the interception time. Mathematically, it is equivalent to
    \begin{align}
        \theta_\mathrm{imp} = \gamma_\mathrm{P_f}-\gamma_\mathrm{T_f},
    \end{align}
    where $\gamma_\mathrm{T_f}$ and $\gamma_\mathrm{P_f}$ are the headings of the target and the interceptor at the interception time, respectively. 
\end{definition}
Without loss of generality, one can consider the target's heading, $\gamma_\mathrm{T_f}$, to be zero in the case of a stationary target, which results in $\theta_\mathrm{imp}=\gamma_\mathrm{P_f}$. When the interceptor is on a collision course, i.e. $r\dot{\theta}=0$, it follows from \eqref{eq: theta_dot}
\begin{align}\label{eq:thetadotdef}
   \dot{\theta} = -\dfrac{v_\mathrm{P}\sin(\gamma_\mathrm{P_f}- \theta_\mathrm{d})}{r} = 0  \implies  -\dfrac{v_\mathrm{P}\sin(\theta_\mathrm{imp}- \theta_\mathrm{d})}{r} = 0.
\end{align}
where $\theta_\mathrm{d}$ is the desired impact angle. Solving \eqref{eq:thetadotdef} over $(-\pi, \pi]$ yields two solutions as $\theta_\mathrm{d}= \theta_\mathrm{imp},~ \pi + \theta_\mathrm{imp}$.

Accordingly, the control objective is to design the interceptor lateral acceleration command such that the interceptor achieves interception at a prescribed impact time and a prescribed impact angle simultaneously.
\begin{problem*}
    For given desired terminal values $t_\mathrm{d}$ and $\theta_\mathrm{d}$, design a nonlinear guidance law (the interceptor's lateral acceleration command) that leads to impact time- and angle-constrained interception of a stationary target, that is, $r(t_\mathrm{d})=0$ and $\theta_\mathrm{imp}(t_\mathrm{d})=\theta_\mathrm{d}$.
\end{problem*}

\section{Impact Time- and Angle-Constrained Guidance Strategy}
From \Cref{def:time}, it follows that the desired time-to-go $t_\mathrm{go}^\mathrm{d} = t_\mathrm{d} - t$, and hence, regulating the impact time is essentially regulating the engagement duration time-to-go. As a representative formulation of the engagement duration, consider the time-to-go estimate against the stationary target, which accounts for larger heading angle errors \cite{9000526}, given by 
\begin{equation} \label{eq: time-to-go expression}
    t_\mathrm{go}  = \dfrac{r}{v_\mathrm{P}} \left( 1 + \dfrac{\sin^2 \sigma_\mathrm{P}}{K}  \right),
\end{equation}
where $K = 4N -2$, with $N \geq3$ denotes the navigation constant in PN guidance. It follows from \eqref{eq: time-to-go expression} that the interception occurs when $t_\mathrm{go}=0$, which corresponds to $r = 0$ (and vice versa). Thus $t_\mathrm{go}=0 \iff r=0$.
\begin{lemma}\label{lemma: dynamics of time-to-go}
    The dynamics of the time-to-go of the interceptor has a relative degree of one with respect to its lateral acceleration, $a_\mathrm{P}$.
\end{lemma}
    \begin{proof}
        Differentiating \eqref{eq: time-to-go expression} with respect to time, one may obtain
\begin{align}\label{eq;tgo_dot12}
    \dot{t}_\mathrm{go}  &= \dfrac{\dot{r}}{v_\mathrm{P}} +  \dfrac{1}{Kv_\mathrm{P}} \left( \dot{r}\sin^2\sigma_\mathrm{P} + 2r \sin \sigma_\mathrm{P} \cos \sigma_\mathrm{P} \dot{\sigma}_\mathrm{P}
    \right). 
\end{align}
Further simplification using the result from \Cref{lemma 2} and substituting \eqref{eq: r_dot} into \eqref{eq;tgo_dot12} yields
\begin{align}\label{eq: time-to-go dynamics}
      \dot{t}_\mathrm{go} 
      &= -\cos \sigma_\mathrm{P} + \dfrac{1}{Kv_\mathrm{P}} \left( 2v_\mathrm{P}\sin^2\sigma_\mathrm{P} \cos \sigma_\mathrm{P} -  v_\mathrm{P}\sin^2\sigma_\mathrm{P} \cos \sigma _\mathrm{P}+ \dfrac{r \sin 2\sigma_\mathrm{P}}{v_\mathrm{P}}a_\mathrm{P}\right), \nonumber  \\
        &= -\cos \sigma_\mathrm{P} \left( 1 - \dfrac{\sin^2 \sigma_\mathrm{P}}{K}\right) + \left( \dfrac{r \sin 2\sigma_\mathrm{P}}{K v^2_\mathrm{P}}\right)a_\mathrm{P}.
       \end{align}
It is evident from \eqref{eq: time-to-go dynamics} that the time-to-go dynamics possess a relative degree of one with respect to the lateral acceleration of the interceptor.
\end{proof}
\Cref{lemma: dynamics of time-to-go} shows that the interceptor’s lateral acceleration directly influences the time-to-go, which is critical for the design of time-constrained guidance law.
\begin{lemma}\label{lemma:losangle}
    The dynamics of the LOS angle has a relative degree of two with respect to the interceptor's lateral acceleration.
\end{lemma}
\begin{proof}
    Differentiating \eqref{eq: theta_dot} with respect to time yields
    \begin{align}\label{eq:losrate}
    \ddot{\theta} &= - \left(\dfrac{v_\mathrm{P}\cos \sigma_\mathrm{P}\dot{\sigma}_\mathrm{P}}{r} - \dfrac{v_\mathrm{P}\sin \sigma_\mathrm{P}\dot{r}}{r^2}\right).   
    \end{align}
Using the results from \Cref{lemma 2} and substituting \eqref{eq: r_dot} into \eqref{eq:losrate} leads to 
\begin{align}\label{eq:losrate_final}
    \ddot{\theta}   &= - \left(\dfrac{v_\mathrm{P}\cos \sigma_\mathrm{P}}{r}\left(\dfrac{a_\mathrm{P}}{v_\mathrm{P}} + \dfrac{v_\mathrm{P}\sin \sigma_\mathrm{P}}{r}\right) + \dfrac{v^2_\mathrm{P}\sin \sigma_\mathrm{P}\cos \sigma_\mathrm{P}}{r^2} \right) \nonumber \\
    &= -\dfrac{v_\mathrm{P}^2 \sin 2\sigma_\mathrm{P}}{r^2} - \dfrac{\cos \sigma_\mathrm{P}}{r}a_\mathrm{P}.
\end{align}
    From \eqref{eq:losrate_final}, it can be observed that the dynamics of the LOS angle is also influenced by the interceptor's lateral acceleration.
\end{proof}
Consider the sub-sliding surface representing the error between the estimated time-to-go and the desired time-to-go, given as
\begin{equation} \label{eq: time-to-go error}
    s_\mathrm{t}    = t_\mathrm{go} - t_\mathrm{go}^\mathrm{d} = t_\mathrm{go} - (t_\mathrm{d} -t),
\end{equation}
where $\dot{t}_\mathrm{go}^\mathrm{d} = -1$. 
\begin{lemma}\label{lem:steq}
    Consider the sub-sliding surface \eqref{eq: time-to-go error}. The equivalent control associated with the manifold $s_\mathrm{t}=0$ is given by
    \begin{align}
        u_\mathrm{t} = - \dfrac{Kv_\mathrm{P}^2 \tan \frac{\sigma_\mathrm{P}}{2}}{2 r \cos \sigma_\mathrm{P}} - \dfrac{v_\mathrm{P}^2 \sin \sigma_\mathrm{P}}{2 r} .\label{eq:ueq1}
    \end{align}
    Under this control, the manifold $s_\mathrm{t}=0$ is rendered invariant, that is, $s_\mathrm{t}=0\implies \dot{s}_\mathrm{t}=0$.
\end{lemma}
\begin{proof}
    On differentiating \eqref{eq: time-to-go error} with respect to time, we obtain
\begin{equation} \label{eq: error dot 1}
    \dot{s}_\mathrm{t}   = \dot{t}_\mathrm{go} + 1.
\end{equation}
Using the results from \Cref{lemma: dynamics of time-to-go}, we obtain the dynamics of the sub-sliding surface $s_\mathrm{t}$ as
\begin{align} \label{eq: error rate for time-to-go}
    \dot{s}_\mathrm{t}  = 1 - \cos \sigma_\mathrm{P} \left( 1 - \dfrac{\sin^2\sigma_\mathrm{P}}{K}
    \right) +  \dfrac{r\sin2\sigma_\mathrm{P}}{Kv_\mathrm{P}^2} a_\mathrm{P}.
\end{align}
By definition, the equivalent control is the continuous control input required to maintain motion on the sliding manifold once the trajectory has reached it. Hence, on the manifold $s_\mathrm{t}=0$, one imposes the invariance condition $\dot{s}_\mathrm{t}=0$. Using the control-affine dynamics of $s_\mathrm{t}$ in \eqref{eq: error rate for time-to-go}, this yields
\begin{align} 
    0  = 1 - \cos \sigma_\mathrm{P} \left( 1 - \dfrac{\sin^2\sigma_\mathrm{P}}{K}
    \right) +  \dfrac{r\sin2\sigma_\mathrm{P}}{Kv_\mathrm{P}^2} a_\mathrm{P},
\end{align}
which leads to the expression in \eqref{eq:ueq1} after simplifying using $\cos \sigma_\mathrm{P}= 1 - 2 \sin^2\dfrac{\sigma_\mathrm{P}}{2}$.
\end{proof}
\begin{remark}
    The quantity \eqref{eq:ueq1} represents the nominal lateral acceleration command required to exactly preserve the time-channel manifold $s_\mathrm{t}$ once it is reached. This essentially means that it characterizes the ideal control action associated with the impact time objective alone.
\end{remark}
In the present problem, however, the same input $a_\mathrm{P}$ must also regulate the impact angle channel. To enforce the impact angle constraints, we first define the impact angle error as 
\begin{equation}\label{eq: Error variable 2}
    e_{\theta}       = \theta - \theta_\mathrm{d}. 
\end{equation}
 Differentiating $e_\theta$ with respect to time yields
\begin{equation}\label{eq: error_2 dynamics}
    \dot{e}_{\theta}  = \dot{\theta} = -\dfrac{v_\mathrm{P}\sin \sigma_\mathrm{P}}{r}.
\end{equation}
The second time derivative of $e_\theta$, together with the results from \Cref{lemma 2}, yields
\begin{equation}
    \ddot{e}_{\theta}  =- \dfrac{v_\mathrm{P}^2 \sin 2\sigma_\mathrm{P}}{r^2} -\dfrac{ \cos \sigma_\mathrm{P}}{r}a_\mathrm{p}. \label{eq: Error variable 2 double derivative}
\end{equation}
Now, consider the second sub-sliding surface corresponding to the impact angle error, given as
\begin{equation}\label{eq: Sliding manifold for LOS case}
    s_{\theta}     = {e}_{\theta} + m\dot{e}^{c_1/c_2}_{\theta},  
\end{equation}
where $m>0$ is a design parameter governing the rate of convergence of $e_\theta$, whereas $c_1$ and $c_2$ are odd positive integers such that  $c_1>c_2$ and $1<c_1/c_2<2$. Such a sub-sliding surface facilitates a finite-time error convergence once sliding mode is enforced. 
\begin{lemma}\label{lem:sthetaeq}
    Consider the sub-sliding surface \eqref{eq: Sliding manifold for LOS case}. The equivalent control associated with the manifold $s_\theta=0$ is given by
    \begin{align}
        u_\theta = \dfrac{rc_2}{m c_1 \cos \sigma_\mathrm{P}} \left[ \left (- \dfrac{v_\mathrm{P} \sin \sigma_\mathrm{P}}{r} \right)^{(2 - c_1/c_2)} - \dfrac{m c_1 v^2_\mathrm{P} \sin 2 \sigma_\mathrm{P}}{r^2 c_2} \right].\label{eq:ueq2}
    \end{align}
    Under this control, the manifold $s_\theta=0$ is rendered invariant, that is, $s_\theta=0\implies \dot{s}_\theta=0$.
\end{lemma}
\begin{proof}
    On differentiating \eqref{eq: Sliding manifold for LOS case} with respect to time, one may obtain
\begin{align}\label{eq: s_theta_dot1}
    \dot{s}_\theta = \dot{e}_\theta + \left(\dfrac{mc_1}{c_2}\right)\dot{e}_\theta^{\left(c_1/c_2-1\right)} \ddot{e}_\theta.
\end{align}
Substituting \eqref{eq: error_2 dynamics} and \eqref{eq: Error variable 2 double derivative} in \eqref{eq: s_theta_dot1} yields
\begin{align}\label{eq: Dynamics of S_theta}
    \dot{s}_\theta  =& -\dfrac{v_\mathrm{P}\sin \sigma_\mathrm{P}}{r} + \dfrac{mc_1}{c_2}\left(-\dfrac{v_\mathrm{P}\sin \sigma_\mathrm{P}}{r}\right)^{(c_1/c_2-1)}\left(- \dfrac{v_\mathrm{P}^2 \sin 2\sigma_\mathrm{P}}{r^2} -\dfrac{ \cos \sigma_\mathrm{P}}{r}a_\mathrm{p}\right) \nonumber \\
    =& - \left(\dfrac{v_\mathrm{P}\sin \sigma_\mathrm{P}}{r} + \dfrac{mc_1v^2_\mathrm{P}\sin2\sigma_\mathrm{P}}{r^2c_2}\left(-\dfrac{v_\mathrm{P}\sin \sigma_\mathrm{P}}{r} \right)^{(c_1/c_2-1)}\right) - \dfrac{mc_1 \cos \sigma_\mathrm{P}}{rc_2}\left(-\dfrac{v_\mathrm{P}\sin \sigma_\mathrm{P}}{r} \right)^{(c_1/c_2-1)}a_\mathrm{P}.
\end{align}
On the manifold $s_\theta=0$, one imposes the invariance condition $\dot{s}_\theta=0$. Using the control-affine dynamics of $s_\theta$ in \eqref{eq: Dynamics of S_theta}, this yields
\begin{align} 
    0  = - \left(\dfrac{v_\mathrm{P}\sin \sigma_\mathrm{P}}{r} + \dfrac{mc_1v^2_\mathrm{P}\sin2\sigma_\mathrm{P}}{r^2c_2}\left(-\dfrac{v_\mathrm{P}\sin \sigma_\mathrm{P}}{r} \right)^{(c_1/c_2-1)}\right) - \dfrac{mc_1 \cos \sigma_\mathrm{P}}{rc_2}\left(-\dfrac{v_\mathrm{P}\sin \sigma_\mathrm{P}}{r} \right)^{(c_1/c_2-1)}a_\mathrm{P},
\end{align}
which leads to the expression in \eqref{eq:ueq2} after simplifying using $\sin 2 \sigma_\mathrm{P}= 4 \sin \dfrac{\sigma_\mathrm{P}}{2}\cos \dfrac{\sigma_\mathrm{P}}{2}\cos \sigma_\mathrm{P}$.
\end{proof}
\begin{remark}
    From \eqref{eq: error rate for time-to-go} and \eqref{eq: Dynamics of S_theta}, it is evident that the dynamics of both sub-sliding surfaces depend on the interceptor's lead angle. Moreover, since the interceptor has access only to a single control input, namely the lateral acceleration $a_\mathrm{P}$, the resulting guidance problem, as seen from \eqref{eq: error rate for time-to-go} and \eqref{eq: Dynamics of S_theta}, is underactuated in nature.
\end{remark}
\begin{remark}\label{rem:nonsingularity}
    The validity of the equivalent-control expressions in \Cref{lem:steq,lem:sthetaeq} requires that the corresponding input gain remain nonzero along the engagement. This essentially means that the equivalent-control construction is well defined only when $\sigma_\mathrm{P}$ avoids the singular configurations associated with $\sigma_\mathrm{P}\in\{0,\pi/2,\pi\}$. The treatment of such isolated cases and the conditions under which they can be excluded or systematically handled will be addressed later to demonstrate that the proposed design is nonsingular throughout.
\end{remark}
The underactuated nature of the problem constitutes difficulty in the guidance design since improving performance with respect to one objective may directly influence the evolution of the other. In other words, the same lateral acceleration designed to stabilize $s_\mathrm{t}$ alone does not guarantee stabilization of $s_\theta$ and vice versa. The challenge is to design a single unified lateral acceleration that can drive both $s_\mathrm{t}$ and $s_\theta$ to $0$ simultaneously. To satisfy the terminal constraints on both time and angle, we consider a composite sliding surface as
\begin{equation} \label{eq: Sliding manifold net}
    s    = \lambda s_\mathrm{t} + s_{\theta},
\end{equation}
where $\lambda$ is a time-varying design parameter. 
\begin{lemma}\label{lem:compositesdot}
    The dynamics of the composite sliding surface possesses a relative degree of one with respect to the interceptor's lateral acceleration.
\end{lemma}
\begin{proof}
Differentiating \eqref{eq: Sliding manifold net} with respect to time yields
\begin{align}\label{eq:sdot}
\dot{s}  = \dot{\lambda}s_\mathrm{t} + \lambda \dot{s}_\mathrm{t} + \dot{s}_\theta. 
\end{align}
On substituting \eqref{eq: error rate for time-to-go} and \eqref{eq: Dynamics of S_theta} in \eqref{eq:sdot}, one may obtain
\begin{align}\label{eq:ssdott}
    \dot{s}  &= \dot{\lambda}s_\mathrm{t} + \lambda \left(1 - \cos \sigma_\mathrm{P} \left( 1 - \dfrac{\sin^2\sigma_\mathrm{P}}{K}
    \right) +  \dfrac{r\sin2\sigma_\mathrm{P}}{Kv_\mathrm{P}^2} a_\mathrm{P} \right)  - \left(\dfrac{v_\mathrm{P}\sin \sigma_\mathrm{P}}{r} + \dfrac{mc_1v^2_\mathrm{P}\sin2\sigma_\mathrm{P}}{r^2c_2}\left(-\dfrac{v_\mathrm{P}\sin \sigma_\mathrm{P}}{r} \right)^{(c_1/c_2-1)}\right)  \nonumber \\
    &- \dfrac{mc_1 \cos \sigma_\mathrm{P}}{rc_2}\left(-\dfrac{v_\mathrm{P}\sin \sigma_\mathrm{P}}{r} \right)^{(c_1/c_2-1)}a_\mathrm{P} ,
\end{align}
which indicates that the dynamics of the composite sliding surface is related to the lateral acceleration command with a relative degree of one. 
\end{proof}
We now present the proposed lateral acceleration command $a_\mathrm{P}$ in the next theorem.
\begin{theorem} \label{theorem 1}
    Consider the interceptor-target engagement kinematics whose relative motion is governed by \eqref{eq:enggeo}, the time-to-go formulation in \eqref{eq: time-to-go expression}, and the impact angle error \eqref{eq: Error variable 2}. The proposed interceptor's lateral acceleration command,
    \begin{align}\label{eq:lateral  acceleration}
        a_\mathrm{P}  =&  \dfrac{rc_2}{m c_1 \cos \sigma_\mathrm{P}} \left[ \left (- \dfrac{v_\mathrm{P} \sin \sigma_\mathrm{P}}{r} \right)^{(2 - c_1/c_2)} - \dfrac{m c_1 v^2_\mathrm{P} \sin 2 \sigma_\mathrm{P}}{r^2 c_2} \right]  - \dfrac{Kv_\mathrm{P}^2 \tan \frac{\sigma_\mathrm{P}}{2}}{2 r \cos \sigma_\mathrm{P}} - \dfrac{v_\mathrm{P}^2 \sin \sigma_\mathrm{P}}{2 r} -\bar{\psi}_1~ \sign(s) - \bar{\psi}_2 s, 
    \end{align}
   where 
   \begin{align}
       \bar{\psi}_1   =& \dfrac{ \psi_1}{\chi_2 + \left(-\frac{mc_1 \cos \sigma_\mathrm{P}}{rc_2}\right)\left(-\frac{v_\mathrm{P}\sin \sigma_\mathrm{P}}{r} \right)^{(c_1/c_2 -1)} + \lambda \left(\chi_1 + \left(\frac{r \sin 2\sigma_\mathrm{P}}{Kv_\mathrm{P}^2} \right) \right)} ,\label{eq:etabar} \\
    \bar{\psi}_2=& \dfrac{\psi_2}{\chi_2 + \left(-\frac{mc_1 \cos \sigma_\mathrm{P}}{rc_2}\right)\left(-\frac{v_\mathrm{P}\sin \sigma_\mathrm{P}}{r} \right)^{(c_1/c_2 -1)} + \lambda \left(\chi_1 + \left(\frac{r \sin 2\sigma_\mathrm{P}}{Kv_\mathrm{P}^2} \right) \right)}.\label{eq:kappabar}
   \end{align}
   are adaptive gains such that the design parameters satisfy $\psi_1,\psi_2>0$, $\chi_1,\chi_2\ge0$, ensures that the target is intercepted at the prescribed impact time and impact angle.
 \end{theorem}
\begin{proof}
   We design the lateral acceleration as a sum of three components:
\begin{equation}\label{eq: a_P form}
    a_\mathrm{P} = u_\mathrm{t} + u_\theta + u_\mathrm{c},
\end{equation}
where $u_\mathrm{t}$ and $u_\theta$ are equivalent control terms corresponding to \eqref{eq: error rate for time-to-go} and \eqref{eq: Dynamics of S_theta}, respectively (see \Cref{lem:steq,lem:sthetaeq}), and $u_\mathrm{c}$ is to be designed as a switching term that ensures sliding mode convergence. 

To derive the switching control law $u_\mathrm{ct}$, consider a Lyapunov function candidate as
\begin{align}\label{eq: Lyapunov candidate function}
     V= \dfrac{s^2}{2}.
\end{align}
Differentiating $V$ with respect to time yields
\begin{align}\label{eq:V dot}
    \dot{V} &= s\dot{s}, \nonumber \\
            &= s\left( \lambda\dot{s}_\mathrm{t} + \dot{\lambda}s_\mathrm{t} +    \dot{s}_\theta \right). 
\end{align}
Using the results in \Cref{lem:compositesdot}, one may write 
\begin{align} \label{eq: V_dot simplification}
   \dot{V} &= s \left[ \dot{\lambda}s_\mathrm{t} + \lambda \left(1 - \cos \sigma_\mathrm{P} \left( 1 - \dfrac{\sin^2\sigma_\mathrm{P}}{K}
    \right) +  \dfrac{r\sin2\sigma_\mathrm{P}}{Kv_\mathrm{P}^2} a_\mathrm{P} \right)  - \left(\dfrac{v_\mathrm{P}\sin \sigma_\mathrm{P}}{r} + \dfrac{mc_1v^2_\mathrm{P}\sin2\sigma_\mathrm{P}}{r^2c_2}\left(-\dfrac{v_\mathrm{P}\sin \sigma_\mathrm{P}}{r} \right)^{(c_1/c_2-1)}\right) \right.  \nonumber \\
      & \left. - \dfrac{mc_1 \cos \sigma_\mathrm{P}}{rc_2}\left(-\dfrac{v_\mathrm{P}\sin \sigma_\mathrm{P}}{r} \right)^{(c_1/c_2-1)}a_\mathrm{P}  \right].
\end{align}
Using \eqref{eq: a_P form} and the results from \Cref{lem:steq} and \Cref{lem:sthetaeq}, one may write 
\begin{align}\label{eq:v_dot SS_dot}
    \dot{V} = &s\left[\lambda\dfrac{r \sin 2 \sigma_\mathrm{P}}{Kv^2_\mathrm{P}} \left(u_\theta + u_\mathrm{ct}\right) + \dot{\lambda}s_\mathrm{t} -\dfrac{mc_1\cos \sigma_\mathrm{P}}{rc_2}\left(-\dfrac{v_\mathrm{P} \sin \sigma_\mathrm{P}}{r}\right)^{(c_1/c_2-1)} \left(u_\mathrm{t} + u_\mathrm{c}\right) \right], \nonumber \\
    =& s \left[ \lambda \dfrac{r \sin 2\sigma_\mathrm{P}}{Kv^2_\mathrm{P}}u_\theta - \dfrac{mc_1\cos \sigma_\mathrm{P}}{rc_2}\left(-\dfrac{v_\mathrm{P} \sin \sigma_\mathrm{P}}{r}\right)^{(c_1/c_2-1)} u_\mathrm{t} \right.\nonumber\\
    &\left.+ \dot{\lambda}s_\mathrm{t} + \left( \lambda \dfrac{r \sin 2\sigma_\mathrm{P}}{Kv^2_\mathrm{P}}    - \dfrac{mc_1\cos \sigma_\mathrm{P}}{rc_2}\left(-\dfrac{v_\mathrm{P}\sin \sigma_\mathrm{P}}{r} \right)^{(c_1/c_2-1)} \right) u_\mathrm{c} \right].
\end{align} 
To ensure the stability of the composite sliding surface, $s$, let the design parameter, $\lambda$, evolve according to the adaptation law \cite{6283395}
\begin{align}\label{eq: adaptive law}
    \dot{\lambda} = -\lambda \dfrac{\left(\chi_1 + \frac{r \sin 2\sigma_\mathrm{P}}{Kv_\mathrm{P}^2}  \right)u_\theta s_\mathrm{t}}{||s_\mathrm{t}||^2 + \Omega} -\dfrac{\chi_2 + \left(-\frac{mc_1 \cos \sigma_\mathrm{P}}{rc_2}\right)\left(-\frac{v_\mathrm{P}\sin \sigma_\mathrm{P}}{r} \right)^{(c_1/c_2 -1)}u_\mathrm{t}s_\mathrm{t}}{||s_\mathrm{t}||^2 + \Omega}
\end{align}
where $\Omega\geq0$ is a design parameter which keeps $\dot{\lambda}$ non-singular. It is apparent from \eqref{eq: adaptive law} that $\lambda$ converges to a constant value as $s_\mathrm{t} \to 0$. Then one has,
\begin{align}
      \dot{\lambda}s_\mathrm{t} &= \dfrac{mc_1\cos \sigma_\mathrm{P}}{rc_2}\left(-\dfrac{v_\mathrm{P} \sin \sigma_\mathrm{P}}{r}\right)^{(c_1/c_2-1)} u_\mathrm{t} -\lambda \dfrac{r \sin 2\sigma_\mathrm{P}}{Kv^2_\mathrm{P}}u_\theta  \label{eq: lambda_dot}
\end{align}
and the switching component is chosen as
\begin{align}
    u_\mathrm{c}  & =  -\bar{\psi}_1~\sign(s) -\bar{\psi}_2s, \label{eq: switching case}
\end{align}
where $\bar{\psi}_1$ and $\bar{\psi}_2$ are given in \eqref{eq:etabar}--\eqref{eq:kappabar}.

On substituting the expressions from \eqref{eq: lambda_dot} and \eqref{eq: switching case} into \eqref{eq:v_dot SS_dot}, one may obtain
\begin{align} \label{eq: Lyapunov Derivative}
    \dot{V}  &= - \psi_1~ s ~\sign(s) - \psi_2 s^2, \nonumber \\
    &= -\psi_1|s| -\psi_2 s^2,
\end{align}
which is negative definite for all $s\neq 0$, thus ensuring the stability of the composite sliding mode dynamics. 

Moreover, integrating \eqref{eq: Lyapunov Derivative} gives
\begin{align}
    & V(t)-V(0) =\int_0^t \left(-\psi_1|s| -\psi_2 s^2\right)dt,
\end{align}
which implies
\begin{align}
    V(0) =& V(t) + \int_0^t \left(\psi_1|s| + \psi_2 s^2\right)dt \geq  \int_0^t \left(\psi_1|s| +\psi_2 s^2\right)dt. \label{eq: Integrate 1}
\end{align}
From \eqref{eq: Integrate 1}, it follows that
\begin{align}\label{eq: lim}
 \lim_{t \to \infty} \int_0^t   \left(\psi_1|s| +\psi_2 s^2\right)dt < \infty.
\end{align}
From \eqref{eq: lim}, we infer that $\lim_{t \to \infty} \int_0^t   \psi_1|s| dt < \infty$ and $\lim_{t \to \infty} \int_0^t   \psi_2 s^2dt < \infty$. This implies $s \in \mathcal{L}_1$ and $s \in \mathcal{L}_2$. 
From \eqref{eq: Lyapunov candidate function} and \eqref{eq: Integrate 1}, we obtain
\begin{align}\label{eq: s bound}
    \dfrac{s^2}{2}=V(t)= V(0) - \int_0^t \left(\psi_1|s| +\psi_2 s^2\right)dt < \infty; \quad  \forall ~t\geq 0,
\end{align}
which implies that $s\in \mathcal{L}_\infty$, thus showing that the composite sliding surface $s$ is bounded. From \eqref{eq: Lyapunov Derivative}, we have
\begin{align}\label{eq: s_dot bound}
    \dfrac{dV}{dt}= s\dfrac{ds}{dt}= -\psi_1|s|-\psi_2 s^2 < \infty \quad \forall ~t\geq 0,
\end{align}
so that $\dot{s} \in \mathcal{L}_\infty$. Therefore, by Barbalat's lemma, the composite sliding surface, $s$, is asymptotically stable.
    
Once the sliding mode is enforced, \eqref{eq: Sliding manifold net} becomes equal to zero. To guarantee target interception at the prescribed impact time and angle, it is essential to ensure that both $s_\mathrm{t}$ and $s_\theta$ become zero thereafter. To this end, we now demonstrate that the sub-sliding surfaces, $s_\mathrm{t}$, and $s_\theta$ are also asymptotically stable after sliding mode is enforced on $s$. The time horizon can be divided into two phases by the time instant  $t_\alpha$. In the interval $[0, t_\alpha]$, the system state trajectories move towards the composite sliding surface, $s$. In the interval $(t_\alpha,  \infty)$, the trajectories stay on the composite sliding surface and converge to the origin.

To examine the stability of the sub-sliding surfaces within the interval $[0, t_\alpha]$, define a truncation operator as 
\begin{equation}\label{eq: function}
   h_\mathrm{t_\alpha}(t)= \left\{
    \begin{aligned}
        &1, \quad \quad t\leq t_\alpha \\
        &0, \quad \quad t > t_\alpha
    \end{aligned}
    \right.
\end{equation}
and define the truncated sub-sliding surface variables as
   \begin{align}
       s_{i_\mathrm{t_\alpha}} = s_{i} h_\mathrm{t_\alpha}(t), \quad \dot{s}_{i_\mathrm{t_\alpha}}= \dot{s}_{i}h_\mathrm{t_\alpha}(t),
   \end{align}
where $i \in \{\mathrm{t}, \theta\}$. With this construction\footnote{Here $\dot s_{i,t_\alpha}$ is defined as above rather than obtained by differentiation. It is to be considered only as an analysis device.}, the analysis is restricted to the interval $[0,t_\alpha]$, over which $s_{i,t_\alpha}$ and $\dot{s}_{i,t_\alpha}$ coincide with $s_i$ and $\dot{s}_i$, respectively. For $t>t_\alpha$, both truncated variables vanish identically. Consequently, $s_{i,t_\alpha}$ and $\dot{s}_{i,t_\alpha}$ remain bounded on $[0,t_\alpha]$, and once sliding mode is enforced at $t_\alpha$, the truncated representation is consistent with convergence to the corresponding sliding manifold.

In the subsequent interval $(t_\alpha,\infty)$, the state trajectory has reached the composite sliding manifold $s$ and the reduced-order dynamics evolve autonomously on this manifold. To characterize the limiting behavior, define the compact positively invariant set
\begin{align}\label{eq: s_w}
S_\mathrm{w} \coloneqq \left\{ s\in\mathbb{R}^2 \middle| \dot V(s)\leq 0 \right\},
\end{align}
and the set
\begin{align}
S_\mathrm{u} \coloneqq \left\{ s\in S_\mathrm{w} \middle| \dot V(s)=0 \right\}.
\end{align}
Since $\dot V\leq 0$ in $S_\mathrm{w}$, LaSalle's invariance principle implies that every trajectory starting in $S_\mathrm{w}$ approaches, as $t\to\infty$, the largest invariant set contained in $S_\mathrm{u}$. It therefore remains to identify this invariant set.

On the composite manifold,
\begin{align}
s=\lambda s_\mathrm{t}+s_\theta.
\end{align}
Hence, any invariant motion in $S_\mathrm{u}$ must satisfy
\begin{align}\label{eq: intersection}
\lambda s_\mathrm{t}+s_\theta=0,~~\dfrac{d}{dt}\left(\lambda s_\mathrm{t}+s_\theta\right)=0.
\end{align}
The only invariant solution consistent with $\dot{V}=0$ is the coordinate origin of the sub-sliding-surface space, namely
\begin{align}\label{eq: final step}
s_\mathrm{t}=0,~~s_\theta=0
\end{align}
Therefore, the largest invariant set contained in $S_\mathrm{u}$ reduces to the origin. By LaSalle's invariance principle, it follows that the sub-sliding surfaces $s_\mathrm{t}$ and $s_\theta$ also converge asymptotically to zero. Hence, the sub-sliding surfaces are asymptotically stable.

After sliding mode is enforced on $s_\mathrm{t}$, the interceptor aligns on the requisite trajectory that would lead to a time-constrained interception since $s_\mathrm{t}=0$ implies that $t_\mathrm{go}=t_\mathrm{go}^\mathrm{d}$ from \eqref{eq: time-to-go error}. Furthermore, after $s_\theta=0$, the impact angle error \eqref{eq: Error variable 2} also converges to zero.
    \end{proof}
\begin{remark}\label{remark}
    If \eqref{eq: final step} does not hold, then the composite sliding surface would converge to a point other than the coordinate origin defined by the axes $s_\mathrm{t}$ and $s_\theta$. In the phase plane spanned by $s_\mathrm{t}$ and $s_\theta$,
    \begin{align}
        \lim_{t \to \infty}s = \beta;~ \beta \in \mathbb{R}\setminus \{0\}.
    \end{align}
This indicates that the composite sliding surface is not asymptotically stable, which contradicts the asymptotic stability of $s$.

\end{remark}
\begin{remark}\label{rem:s_rel}
    Following \Cref{lem:steq} and \eqref{eq: a_P form}, one may write
    \begin{align} 
    \dot{s}_\mathrm{t}  = 1 - \cos \sigma_\mathrm{P} \left( 1 - \dfrac{\sin^2\sigma_\mathrm{P}}{K}
    \right) +  \dfrac{r\sin2\sigma_\mathrm{P}}{Kv_\mathrm{P}^2} \left(u_\mathrm{t} + u_\theta + u_\mathrm{c}\right),
\end{align}
which, after the composite surface is reached, is equivalent to
\begin{align} 
    \dot{s}_\mathrm{t}  = \dfrac{r\sin2\sigma_\mathrm{P}}{Kv_\mathrm{P}^2}  u_\theta,
\end{align}
since $u_\mathrm{c}=0$ in the ideal sliding-mode sense. If we let $W=\dfrac{1}{2}s_\mathrm{t}^2$, then it follows that
\begin{align}
    s_\mathrm{t}\dot{s}_\mathrm{t} = s_\mathrm{t}\dfrac{r\sin2\sigma_\mathrm{P}}{Kv_\mathrm{P}^2}  u_\theta.
\end{align}
To guarantee that $s_\mathrm{t}\to 0$, one may impose a sufficient condition on the reduced-order dynamics such that
\begin{align}
    s_\mathrm{t}\dot{s}_\mathrm{t} \leq - \varrho_\mathrm{t}s_\mathrm{t}^2 = -2\varrho_\mathrm{t}W;~~\forall~\varrho_\mathrm{t}>0,
\end{align}
so one has
\begin{align}
    W(t) \leq W(t_r)\exp \left\{-2\varrho_\mathrm{t}\left(t-t_r\right)\right\}.
\end{align}
The above result implies that $\lim_{t\to\infty} s_\mathrm{t}\to 0$. On the composite manifold, from \eqref{eq: intersection}, $s_\theta = -\lambda s_\mathrm{t}$. Assume that $\exists~\bar{\lambda}>0$ such that $\sup_{t\geq t_r}|\lambda(t)|\leq \bar{\lambda}$. Then, 
\begin{align}
    |s_\theta(t)| \leq |\lambda(t)|\, |s_\mathrm{t}(t)| \leq \bar{\lambda}|s_\mathrm{t}(t)|.
\end{align}
Since $s_\mathrm{t}\to 0$, it follows immediately that $s_\theta\to 0$. Thus, the proposed strategy allows for time correction first, and thereafter angle correction is effected once $s_\theta$ (and therefore $e_\theta$) converges to zero.
\end{remark}
Recall \Cref{rem:nonsingularity} that the equivalent-control expressions remain well defined only when $\sigma_\mathrm{P}\notin \{0,\pi/2,\pi\}$. This motivates the following analysis of equilibrium points of the lead angle dynamics and nonsingularity conditions, through which the admissible operating region of the proposed guidance command in \Cref{theorem 1} is characterized.
\begin{theorem}\label{thm:leadangle}
Consider the interceptor-target engagement kinematics whose relative motion is governed by \eqref{eq:enggeo}, the time-to-go formulation in \eqref{eq: time-to-go expression}, and the impact angle error \eqref{eq: Error variable 2}. Under the proposed interceptor's lateral acceleration command \eqref{eq:lateral  acceleration}, the lead angle dynamics admits no equilibrium at $\sigma_\mathrm{P}=\pm\dfrac{\pi}{2}$. Moreover, the configuration $\sigma_\mathrm{P}=\pi$ corresponds to inverse collision course, whereas $\sigma_\mathrm{P}=0$ can occur only at the final impact time if $\sigma_\mathrm{P}(0)\neq 0$. 
\end{theorem}
\begin{proof}
Recall the lead angle dynamics from \Cref{lemma 2} and note that $\sigma_\mathrm{P}=\pi$ corresponds to an inverse-collision geometry and is incompatible with the interception scenario of interest. Hence, this configuration is excluded from the admissible operating region.

Next, consider $\sigma_\mathrm{P}(0)= 0$. If the prescribed terminal conditions satisfy $t_\mathrm{f}=r(0)/v_\mathrm{P}$ and $\theta_\mathrm{d}=0$, then the interceptor is already on the required collision course and no maneuver is needed, so that $\sigma_\mathrm{P}(t)\equiv 0$ and $a_\mathrm{P}(t)\equiv 0$. This is a degenerate but nonsingular case. Excluding this trivial situation, if $t_\mathrm{f}\neq r(0)/v_\mathrm{P}$ and $\sigma_\mathrm{P}(0)\neq 0$, then $\sigma_\mathrm{P}=0$ cannot occur for any $t<t_\mathrm{f}$. Otherwise, the interceptor would already be on the final collision course before satisfying the full terminal requirements. The same argument holds even if  $\sigma_\mathrm{P}(0)= 0$. In this case, since $t_\mathrm{f}\neq r(0)/v_\mathrm{P}$, the lead angle will immediately deviate from this zero initial value to place the interceptor on the requisite trajectory.

Consequently, we focus our attention on the steady-state, i.e., after the respective errors have converged. Hence, for nontrivial engagements, $\sigma_\mathrm{P}=0$ is attained only at interception. This may be confirmed by obtaining a relation between the range-to-go and the lead angle, assuming a small lead angle.

Substituting \eqref{eq:lateral  acceleration}  into \eqref{eq: sigma_dot_expression} leads to
\begin{align}\label{eq: sigma_dot_2}
    \dot{\sigma}_\mathrm{P} =   & \dfrac{rc_2}{mv_\mathrm{P} c_1 \cos \sigma_\mathrm{P}} \left[ \left (- \dfrac{v_\mathrm{P} \sin \sigma_\mathrm{P}}{r} \right)^{(2 - c_1/c_2)} - \dfrac{m c_1 v^2_\mathrm{P} \sin 2 \sigma_\mathrm{P}}{r^2 c_2} \right]  - \dfrac{K v_\mathrm{P} \tan \frac{\sigma_\mathrm{P}}{2}}{2 r \cos \sigma_\mathrm{P}} - \dfrac{v_\mathrm{P} \sin \sigma_\mathrm{P}}{2 r}  + \dfrac{v_\mathrm{P} \sin \sigma_\mathrm{P}}{r}.
\end{align}
Using \eqref{eq: r_dot} and \eqref{eq: sigma_dot_2}, one may write
\begin{align}\label{eq: rvssigma_dot2}
    \dfrac{d r}{d \sigma_\mathrm{P}}&= \dfrac{-v_\mathrm{P} \cos \sigma_\mathrm{P}}{\left(- \dfrac{v_\mathrm{P} \sin \sigma_\mathrm{P}}{r}\right)^{(2- c_1/c_2)} \left(\dfrac{rc_2}{mv_\mathrm{P}c_1 \cos \sigma_\mathrm{P}}\right) -\dfrac{2 v_\mathrm{P} \sin \sigma_\mathrm{P}}{r} - \dfrac{Kv_\mathrm{P} \tan \dfrac{\sigma_\mathrm{P}}{2}}{2r \cos \sigma_\mathrm{P}} - \dfrac{v_\mathrm{P}\sin \sigma_\mathrm{P}}{2r} + \dfrac{v_\mathrm{P} \sin \sigma_\mathrm{P}}{r}} \nonumber \\
    &= \dfrac{-v_\mathrm{P} \cos \sigma_\mathrm{P}}{\left(- \dfrac{v_\mathrm{P} \sin \sigma_\mathrm{P}}{r}\right)^{(2- c_1/c_2)} \left(\dfrac{rc_2}{mv_\mathrm{P}c_1 \cos \sigma_\mathrm{P}}\right) - \dfrac{Kv_\mathrm{P} \tan \dfrac{\sigma_\mathrm{P}}{2}}{2r \cos \sigma_\mathrm{P}} - \dfrac{3v_\mathrm{P} \sin \sigma_\mathrm{P}}{2r}}.
\end{align}
In the terminal phase (near interception), the interceptor's lead angle is typically small because its velocity is nearly aligned with the LOS angle. Then, \eqref{eq: rvssigma_dot2} can be simplified further as
\begin{align}\label{eq: rvssigma_dot2_simplified}
    \dfrac{d r}{d \sigma_\mathrm{P}}=&~ \dfrac{-v_\mathrm{P}}{\left(- \dfrac{v_\mathrm{P}\sigma_\mathrm{P}}{r}\right)^{(2- c_1/c_2)} \left(\dfrac{rc_2}{mv_\mathrm{P}c_1 }\right)  - \dfrac{Kv_\mathrm{P} \sigma_\mathrm{P}}{4r} - \dfrac{3v_\mathrm{P} \sigma_\mathrm{P}}{2r}} \nonumber \\
     =&~ \dfrac{-v_\mathrm{P}}{\left(\dfrac{v_\mathrm{P}c_2 \sigma_\mathrm{P}^2}{mrc_1} \right)\left(\dfrac{-r}{v_\mathrm{P} \sigma_\mathrm{P}} \right)^{c_1/c_2}- \left(\dfrac{v_\mathrm{P}\sigma_\mathrm{P}}{r} \right)\left(\dfrac{K}{4} + \dfrac{3}{2} \right)}
\end{align}
under the assumption of a small lead angle in the terminal phase.
Upon taking the reciprocal of \eqref{eq: rvssigma_dot2_simplified} and after further simplifications, one has
\begin{align}\label{eq:rvssigma_dot2}
    \dfrac{d \sigma_\mathrm{P}}{dr} &=~ \dfrac{\sigma_\mathrm{P}}{r}\left(\dfrac{K}{4} + \dfrac{3}{2} \right) - \dfrac{c_2\sigma_\mathrm{P}^2}{mrc_1}\left(\dfrac{-r}{v_\mathrm{P}\sigma_\mathrm{P}}\right)^{c_1/c_2},
    \end{align}
which, upon integrating from $r(0)$ to $r$ (with corresponding $\sigma_\mathrm{P}(0)$ to $\sigma_\mathrm{P}$) yields
\begin{align}\label{eq: sigmaVsrange}
    \sigma_\mathrm{P} = r^{\left(\frac{K}{4} + \frac{3}{2}\right)} \left( \dfrac{\sigma_\mathrm{P}(0)}{r(0)^{\left(\frac{K}{4}+ \frac{3}{2}\right)}}-\dfrac{c_1 -c_2}{mc_1v_\mathrm{p}^{c_1/c_2}}\dfrac{r(0)^{1 - \alpha}}{1-\alpha} + \dfrac{c_1 - c_2}{mc_1v_\mathrm{P}^{c_1/c_2}} \dfrac{r^{1 - \alpha}}{1 - \alpha}\right),
\end{align}
where  $\alpha = \dfrac{(c_1 - c_2)(K - 2)}{4c_2}$. Since $r(0)\neq 0$,  $v_\mathrm{P}\neq 0$, $\alpha\neq 1$, and the design parameters $m,c_1,c_2$ are also nonzero, it is now evident from \eqref{eq: sigmaVsrange} that the lead angle vanishes only at the time of interception (i.e., when $r\to 0$), and hence $a_\mathrm{P}$ remains nonsingular in the steady-state.

It remains to exclude $\sigma_\mathrm{P}=\pm \pi/2$. For this purpose, suppose, for contradiction, that $\sigma_\mathrm{P}= \dfrac{\pi}{2}$ is an equilibrium point, which means once the lead angle reaches $\dfrac{\pi}{2}$, it stays there, thereby $\dot{\sigma}_\mathrm{P}=0$. Then it follows from \eqref{eq: sigma_dot_2} that 
\begin{align}\label{eq: sigma_dot_3}
     K v_\mathrm{P}\tan \frac{\sigma_\mathrm{P}}{2} = \dfrac{mc_1v^2_\mathrm{P} \sin \sigma_\mathrm{P} \cos \sigma_\mathrm{P} +2r^2 c_2\left[ \left (- \dfrac{v_\mathrm{P} \sin \sigma_\mathrm{P}}{r} \right)^{(2 - c_1/c_2)} - \dfrac{m c_1 v^2_\mathrm{P} \sin 2 \sigma_\mathrm{P}}{r^2 c_2} \right]}{mc_1v_\mathrm{P} }.
\end{align}
If $\sigma_\mathrm{P}= \dfrac{\pi}{2}$, then \eqref{eq: sigma_dot_3} reduces to
\begin{align}\label{eq:simp}
         K v_\mathrm{P} &= \dfrac{2r^2c_2}{mc_1}\left(-\dfrac{v_\mathrm{P}}{r}\right)^{2 - c_1/c_2}. \nonumber \\
    \end{align}
     Solving \eqref{eq:simp} for $r$ yields
    \begin{align}\label{eq:steady}
        r  = \left[\dfrac{Km c_1}{2c_2}\left(-v_\mathrm{P}\right)^{c_1/c_2} \right]^{c_2/c_1}.
    \end{align}
    In \eqref{eq:steady}, the LHS varies with time, while the RHS remains constant, and thus the equality cannot be satisfied. Hence, $\sigma_\mathrm{P}=\dfrac{\pi}{2}$ is not an equilibrium point. The case for $\sigma_\mathrm{P}=-\dfrac{\pi}{2}$ can be argued similarly. 
  \end{proof}  
It can be inferred from \Cref{thm:leadangle} that the proposed guidance command, \eqref{eq:lateral  acceleration}, is nonsingular by design everywhere, except in the situation when the trajectories cross isolated points $\sigma_\mathrm{P}=\pm \dfrac{\pi}{2}$ in the transient phase. Since the interceptor's lateral acceleration remains bounded in practice, such isolated points, even if encountered, do not pose any implementation issues. \Cref{fig:rangeVslead} shows a typical $r-\sigma_\mathrm{P}$ plane where one may observe that the $\sigma_\mathrm{P}\to 0$ as $r\to 0$ for any nontrivial case.
\begin{figure}[h!]
    \centering
    \includegraphics[width=0.5\linewidth]{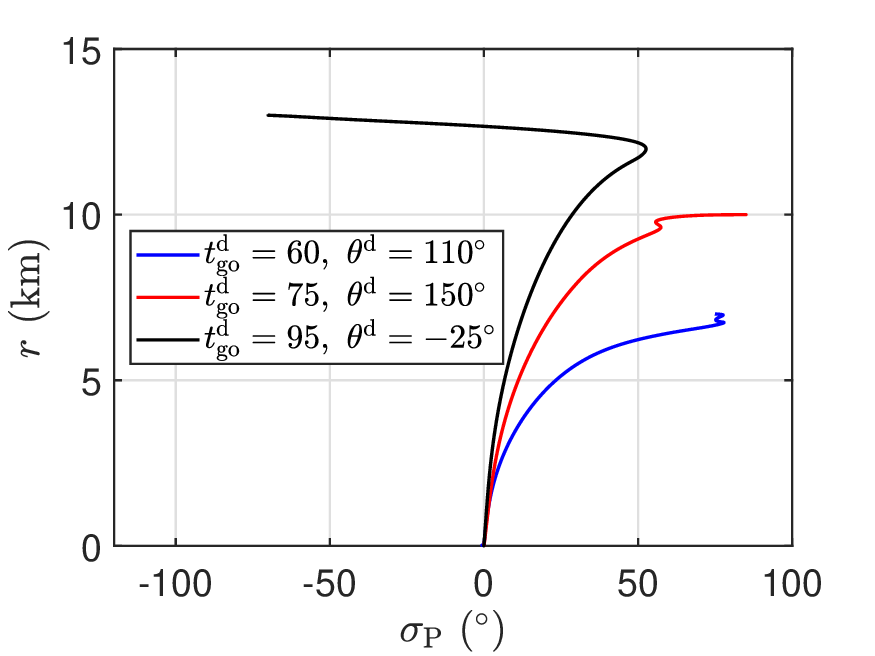}
    \caption{Range versus lead angle for different initial and terminal conditions showing that the lead angle vanishes at the impact time only.}
    \label{fig:rangeVslead}
\end{figure}

To demonstrate the generality of the proposed framework, we redesign the guidance law using a different time-to-go expression and show that the proposed approach retains its effectiveness under this alternative formulation. Let the time-to-go be chosen as \cite{doi:10.1177/09544100211029817}
\begin{align}\label{eq:tgo_2}
    t_\mathrm{go} = \dfrac{r}{v_\mathrm{P}}\left(1+ \dfrac{\sigma_\mathrm{P}^2 + \sigma^2_\mathrm{P_f}}{15} - \dfrac{\sigma_\mathrm{P}\sigma_\mathrm{P_f}}{30}\right),
\end{align}
where
\begin{align}\label{eq:sigma_final_tgo2}
    \sigma_\mathrm{P_f}= \gamma_\mathrm{P_f}-\theta,
\end{align}
is the lead angle at the impact, and $\gamma_\mathrm{P_{f}}$ is the interceptor's heading at that time.
\begin{lemma}\label{lem:tgo_2}
    The dynamics of the interceptor's time-to-go in \eqref{eq:tgo_2} has a relative degree of one with respect to its lateral acceleration.
    \end{lemma}
    \begin{proof}
        Differentiating the time-to-go estimate, \eqref{eq:tgo_2}, with respect to time, yields
       \begin{align}\label{eq:tgo_2_dot1}
    \dot{t}_\mathrm{go} = \dfrac{\dot{r}}{v_\mathrm{P}} \left( 1 + \dfrac{\sigma_\mathrm{P}^2 + \sigma_{\mathrm{P}_\mathrm{f}}^2}{15} - \dfrac{\sigma_\mathrm{P}\sigma_{\mathrm{P}_\mathrm{f}}}{30} \right) + \dfrac{r}{v_\mathrm{P}} \left( \dfrac{2}{15} \left( \sigma_\mathrm{P}\dot{\sigma}_\mathrm{P} + \sigma_{\mathrm{P}_\mathrm{f}}\dot{\sigma}_{\mathrm{P}_\mathrm{f}} \right) - \dfrac{1}{30} \left( \dot{\sigma}_\mathrm{P}\sigma_{\mathrm{P}_\mathrm{f}} + \sigma_\mathrm{P}\dot{\sigma}_{\mathrm{P}_\mathrm{f}} \right) \right).
\end{align}
        Using the results from \Cref{lemma 2} and substituting for the term, $\dot{\sigma}_\mathrm{P_f}$ by differentiating \eqref{eq:sigma_final_tgo2} with respect to time, \eqref{eq:tgo_2_dot1} can be simplified as 
        \begin{align}
            \dot{t}_\mathrm{go} = \dfrac{\dot{r}}{v_\mathrm{P}}\left(  1+ \dfrac{\sigma_\mathrm{P}^2 + \sigma^2_\mathrm{P_f}}{15} - \dfrac{\sigma_\mathrm{P}\sigma_\mathrm{P_f}}{30} \right) + \dfrac{r}{v_\mathrm{P}}\left(\dfrac{2\sigma_\mathrm{P}}{15v_\mathrm{P}}a_\mathrm{P} -\dfrac{2\sigma_\mathrm{P}}{15}\dot{\theta}  -\dfrac{2\sigma_\mathrm{P_f}}{15}\dot{\theta} +\dfrac{\sigma_\mathrm{P}}{30}\dot{\theta} - \dfrac{\sigma_\mathrm{P_f}}{30v_\mathrm{P}}a_\mathrm{P} + \dfrac{\sigma_\mathrm{P_f}}{30}\dot{\theta}\right).
        \end{align}
    On rearranging the terms above yields
    \begin{align}\label{eq:tgo_2_dot2}
        \dot{t}_\mathrm{go} = \dfrac{\dot{r}}{v_\mathrm{P}}\left(1+ \dfrac{\sigma_\mathrm{P}^2 + \sigma^2_\mathrm{P_f}}{15} - \dfrac{\sigma_\mathrm{P}\sigma_\mathrm{P_f}}{30}\right) - \dfrac{r \dot{\theta}}{10v_\mathrm{P}} \left(\sigma_\mathrm{P} + \sigma_\mathrm{P_f} \right) +  \dfrac{r}{v_\mathrm{P}^2} \left(\dfrac{2\sigma_\mathrm{P}}{15} - \dfrac{\sigma_\mathrm{P_f}}{30} \right)a_\mathrm{P}.
    \end{align}
    Substituting the range rate and LOS rate from \eqref{eq: r_dot} and \eqref{eq: theta_dot}, \eqref{eq:tgo_2_dot2} can be further simplified to 
    \begin{align}\label{eq:tgo_2_dot_fin}
        \dot{t}_\mathrm{go} = -\cos \sigma_\mathrm{P}\left(1+ \dfrac{\sigma_\mathrm{P}^2 + \sigma^2_\mathrm{P_f}}{15} - \dfrac{\sigma_\mathrm{P}\sigma_\mathrm{P_f}}{30}\right) + \dfrac{\sin \sigma_\mathrm{P}\left(\sigma_\mathrm{P} + \sigma_\mathrm{P_f} \right)}{10} +  \dfrac{r}{v_\mathrm{P}^2} \left(\dfrac{2\sigma_\mathrm{P}}{15} - \dfrac{\sigma_\mathrm{P_f}}{30} \right)a_\mathrm{P}.
    \end{align}
    It can be observed from \eqref{eq:tgo_2_dot_fin} that this time-to-go dynamics also possesses a relative degree of one with respect to the interceptor's lateral acceleration.
    \end{proof}
\Cref{lem:tgo_2} reveals that the lateral acceleration design corresponding to the time-to-go expression in \eqref{eq:tgo_2} follows directly along the same lines as the previous development, because \eqref{eq: time-to-go expression} and \eqref{eq:tgo_2} share the same relative degree. Therefore, we present the following results without proof.

\begin{lemma}\label{lem:stgo2}
    Consider the sub-sliding surface defined in \eqref{eq: time-to-go error}. Suppose the time-to-go dynamics induced by the chosen time-to-go expression is given by \eqref{eq:tgo_2_dot_fin}. Then, the equivalent control associated with the manifold $s_\mathrm{t}=0$ is given by
    \begin{align}\label{eq:utgo2eq}
      u_\mathrm{t}  &= - \dfrac{\left(1 -\cos \sigma_\mathrm{P}\left(1+ \dfrac{\sigma_\mathrm{P}^2 + \sigma^2_\mathrm{P_f}}{15} - \dfrac{\sigma_\mathrm{P}\sigma_\mathrm{P_f}}{30}\right) + \dfrac{\sin \sigma_\mathrm{P}\left(\sigma_\mathrm{P} + \sigma_\mathrm{P_f} \right)}{10}\right)30v_\mathrm{P}^2}{r\left(4\sigma_\mathrm{P} - \sigma_\mathrm{P_f} \right)}.
    \end{align}
    Under this control, the manifold $s_\mathrm{t}=0$ remains invariant, which implies $\dot{s}_\mathrm{t}=0$.
\end{lemma}
\Cref{lem:sthetaeq} still holds since the impact angle error and its definition remain unchanged under the adoption of a different time-to-go dynamics. Hence, we present the modified lateral acceleration directly in the next theorem.

\begin{theorem} \label{thm:ap2}
     Consider the interceptor-target engagement kinematics whose relative motion is governed by \eqref{eq:enggeo}, the time-to-go formulation in \eqref{eq:tgo_2}, and the impact angle error \eqref{eq: Error variable 2}. The proposed interceptor's lateral acceleration command,
 \begin{align}
     a_\mathrm{P} &= - \dfrac{\left(1 -\cos \sigma_\mathrm{P}\left(1+ \dfrac{\sigma_\mathrm{P}^2 + \sigma^2_\mathrm{P_f}}{15} - \dfrac{\sigma_\mathrm{P}\sigma_\mathrm{P_f}}{30}\right) + \dfrac{\sin \sigma_\mathrm{P}\left(\sigma_\mathrm{P} + \sigma_\mathrm{P_f} \right)}{10}\right)30v_\mathrm{P}^2}{r\left(4\sigma_\mathrm{P} - \sigma_\mathrm{P_f} \right)} - \bar{\psi}_1\sign(s) -\bar{\psi}_2s \nonumber \\ &- \dfrac{rc_2}{mc_1 \cos \sigma_\mathrm{P}\left(-\dfrac{v_\mathrm{P}\sin \sigma_\mathrm{P}}{r}\right)^{(c_1/c_2 -1)}}\left[\dfrac{v_\mathrm{P}\sin \sigma_\mathrm{P}}{r}+ \dfrac{mc_1 v^2_\mathrm{P}\sin 2 \sigma_\mathrm{P}}{r^2c_2} \left(-\dfrac{v_\mathrm{P}\sin \sigma_\mathrm{P}}{r}\right)^{(c_1/c_2 -1)}\right],\label{eq:aP2}
 \end{align}
 where
 \begin{align}
     \bar{\psi}_1 =\dfrac{\psi_1 }{\chi_2 + \left(-\frac{mc_1 \cos \sigma_\mathrm{P}}{rc_2}\right)\left(-\frac{v_\mathrm{P}\sin \sigma_\mathrm{P}}{r} \right)^{(c_1/c_2 -1)} + \lambda \left(\chi_1 +\dfrac{\lambda r\left(4\sigma_\mathrm{P}-\sigma_\mathrm{P_f}\right)}{30v_\mathrm{P}^2} \right)}, \label{eq:psi1bar} \\
      \bar{\psi}_2 =\dfrac{\psi_2}{\chi_2 + \left(-\frac{mc_1 \cos \sigma_\mathrm{P}}{rc_2}\right)\left(-\frac{v_\mathrm{P}\sin \sigma_\mathrm{P}}{r} \right)^{(c_1/c_2 -1)} + \lambda \left(\chi_1 +\dfrac{\lambda r\left(4\sigma_\mathrm{P}-\sigma_\mathrm{P_f}\right)}{30v_\mathrm{P}^2} \right)}, \label{eq:psi2bar} 
 \end{align}
   are adaptive gains such that the design parameters satisfy $\psi_1,\psi_2>0$, $\chi_1,\chi_2\ge0$, ensures that the target is intercepted at the prescribed impact time and impact angle.
 \end{theorem}
\begin{proof}
   Similar to the proof of \Cref{theorem 1}, one has
\begin{equation}\label{eq:a_Pform_tgo2}
    a_\mathrm{P} = u_\mathrm{t} + u_\theta + u_\mathrm{c},
\end{equation}
where $u_\mathrm{t}$ is obtained from \Cref{lem:stgo2}, and $u_\theta$ is the same as given in \Cref{lem:sthetaeq}. The corrective term $u_\mathrm{c}$ is also designed to have the same structure as in the proof of \Cref{theorem 1}. Thereafter, one may choose the same Lyapunov function candidate \eqref{eq: Lyapunov candidate function} and follow the procedures outlined in the proof of \Cref{theorem 1} to obtain the modified command \eqref{eq:aP2} under the time-to-go \eqref{eq:tgo_2}.
\end{proof}

The proposed strategy may be extended for the impact time- and angle-constrained interception of a moving non-maneuvering target by leveraging the concept of predicted interception point (PIP) \cite{doi:10.2514/1.G007122}. This point is a virtual or predicted point at which the target is expected to be intercepted. The interceptor steers its heading towards the PIP, where it perceives the moving target as stationary. If the actual position of the target coincides with the PIP at the desired impact time, then the target is guaranteed to be captured. In \Cref{fig:PIP_figure}, the actual relative range between the target and the interceptor is represented by $r$. However, the interceptor aims to capture the moving target at P, whose relative range is given by $r_\mathrm{P}$, and the corresponding LOS angle in this direction is predicted LOS. The position of the target at PIP for the desired impact time $t_\mathrm{d}$ can be written as
 \begin{subequations}
    \begin{align} \label{eq: PIP coordinates}
        x_\mathrm{T}(t_\mathrm{d}) =  x_\mathrm{T}(t) + v_\mathrm{T}\cos{\gamma_\mathrm{T}}(t_\mathrm{go}),\\
        y_\mathrm{T}(t_\mathrm{d}) =  y_\mathrm{T}(t) + v_\mathrm{T}\sin{\gamma_\mathrm{T}}(t_\mathrm{go}),
    \end{align}
\end{subequations}
where $\left(x_\mathrm{T}(t), y_\mathrm{T}(t)\right)$ and $\left(x_\mathrm{T}(t_\mathrm{d}), y_\mathrm{T}(t_\mathrm{d})\right)$ represent the target coordinates at time instants $t$ and $t_\mathrm{d}$, and $\sigma_\mathrm{T}=\gamma_\mathrm{T}-\theta$ is the lead angle of the target. 
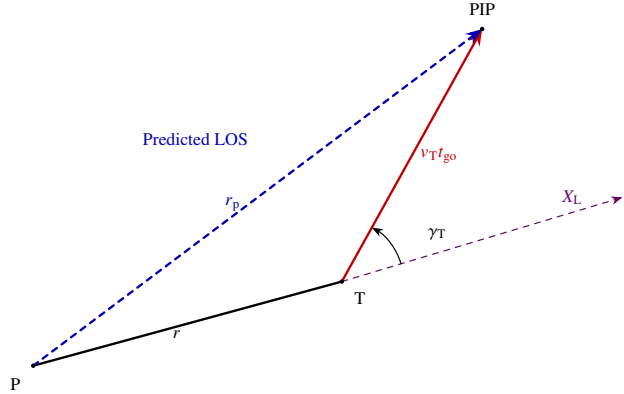
\begin{figure}[h!] 
\centering
\resizebox{.5\linewidth}{!}{%
\begin{tikzpicture}[
    scale=5,
    >=Stealth,
    line cap=round,
    line join=round,
    every node/.style={font=\small},
    los/.style={very thick},
    tgtvel/.style={very thick, red!75!black, -Stealth},
    predlos/.style={very thick, dashed, blue!70!black, -Stealth},
    extline/.style={semithick, dashed, violet!80!black, -Stealth},
    point/.style={circle, fill=black, inner sep=0.8pt},
    lab/.style={fill=none, inner sep=1.2pt, outer sep=1pt}
]

\coordinate (P0) at (0,0);
\coordinate (T)  at (1.10,0.30);
\coordinate (Pp) at (1.60,1.20);
\coordinate (XL) at ($(T)+(1.00,0.30)$);

\draw[los]     (P0) -- (T)  node[midway, below left=1pt, lab] {$r$};
\draw[tgtvel]  (T)  -- (Pp) node[midway, right=2pt, lab] {$v_\mathrm{T}t_\mathrm{go}$};
\draw[predlos] (P0) -- (Pp) node[pos=0.48, left=2pt, lab] {$r_\mathrm{p}$};
\draw[extline] (T)  -- (XL) node[pos=0.82, above=2pt, lab] {$X_\mathrm{L}$};

\node[point, label={[xshift=-2pt,yshift=-2pt]below left:{$\mathrm{P}$}}] at (P0) {};
\node[point, label={[xshift=2pt,yshift=-1pt]below right:{$\mathrm{T}$}}] at (T) {};
\node[point, label={[yshift=2pt]above:{PIP}}] at (Pp) {};

\draw[->, semithick]
    let
        \p1 = ($(XL)-(T)$),
        \p2 = ($(Pp)-(T)$),
        \n1 = {atan2(\y1,\x1)},
        \n2 = {atan2(\y2,\x2)}
    in
    (T) ++(\n1:0.22) arc[start angle=\n1, end angle=\n2, radius=0.22];

\node[lab] at ($(T)+(0.34,0.17)$) {$\gamma_\mathrm{T}$};

\node[lab, text=blue!70!black] at ($(P0)!0.53!(Pp)+(-0.27,0.17)$) {Predicted LOS};

\end{tikzpicture}
}
\caption{Illustration of the PIP for interception of a moving target.}\label{fig:PIP_figure}
\end{figure}

\section{Simulation Results}

We now demonstrate the performance of the proposed approach via simulations. Without loss of generality, we assume that the heading of the target is $0^\circ$. The interceptor's speed remains fixed at $150$ m/s, and the lateral acceleration bound on the interceptor is $|20|$ g, where g is the acceleration due to gravity. In all trajectory plots, the initial positions of the target and the interceptor are depicted by red and black diamond markers  $(\diamond)$, respectively. The navigation constant $N$ is selected as $3$, thus $K=10$.

The effectiveness of the proposed framework against a stationary target is shown in \Cref{fig:Results_com} for various different initial and terminal conditions. The trajectories of the interceptor are depicted in \Cref{fig:traj_com}. It can be observed that by adopting the guidance command \eqref{eq:lateral  acceleration}, the interceptor captures the target successfully in all cases. \Cref{fig:tgo_com} illustrates that the interceptor follows the required trajectories to intercept the target at the desired impact time. The time-to-go profiles exhibit initial deviation, which arises from their trajectories initially deviating and subsequently converging toward the desired path (as shown in \Cref{fig:traj_com}). This signifies that the interceptor takes a detour necessary to attain the requisite time-constrained geometry. The interceptor's lateral acceleration profiles are demonstrated in \Cref{fig:aP_com}. It is apparent from \Cref{fig:aP_com} that during the transient phase, the interceptor requires high lateral acceleration to achieve course correction, followed by a smooth convergence to zero in the terminal phase. As shown in \Cref{fig:snet_com}, the LOS angles converge to their desired values in the terminal phase, whereas sliding mode is enforced on the composite sliding surfaces within around $20$ s in all cases.  The lead angles in \Cref{fig:sig_com} show initial variations due to initial detour, but then decrease to zero monotonically after sliding mode is enforced. This is consistent with the results established in \Cref{thm:leadangle}. The range profiles are demonstrated in \Cref{fig:r_com}. It is apparent that the individual ranges converge to zero at the desired impact time, thereby satisfying the impact-time constraint. Note that both range and the lead angle are zero at the impact time.

\begin{figure}[ht!]
  \centering
  
  \begin{subfigure}[t]{.49\linewidth}
    \centering
    \includegraphics[width=\linewidth]{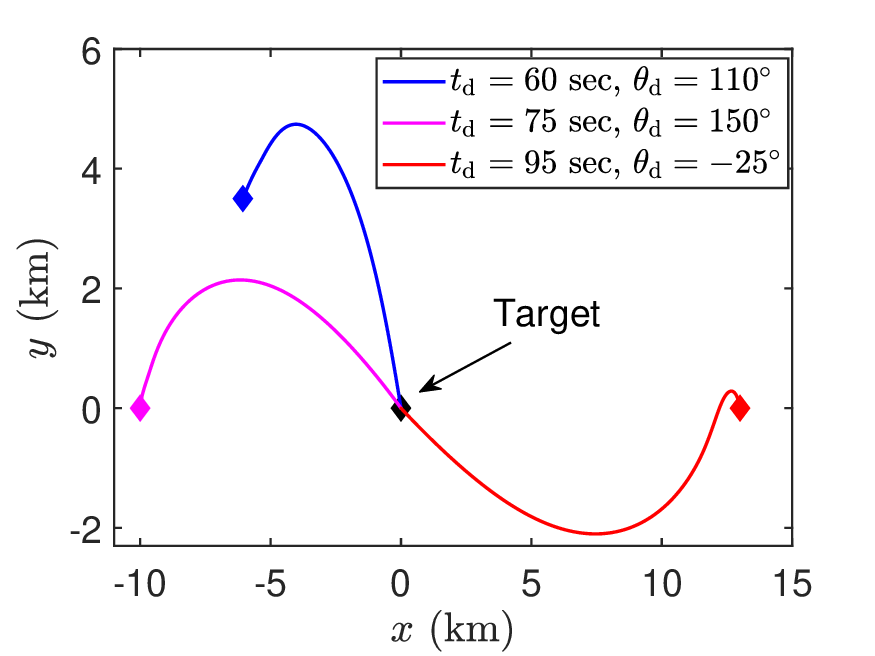}
    \caption{Interceptor's trajectories.}
    \label{fig:traj_com}
  \end{subfigure} 
  \begin{subfigure}[t]{0.49\linewidth}
    \centering
    \includegraphics[width=\linewidth]{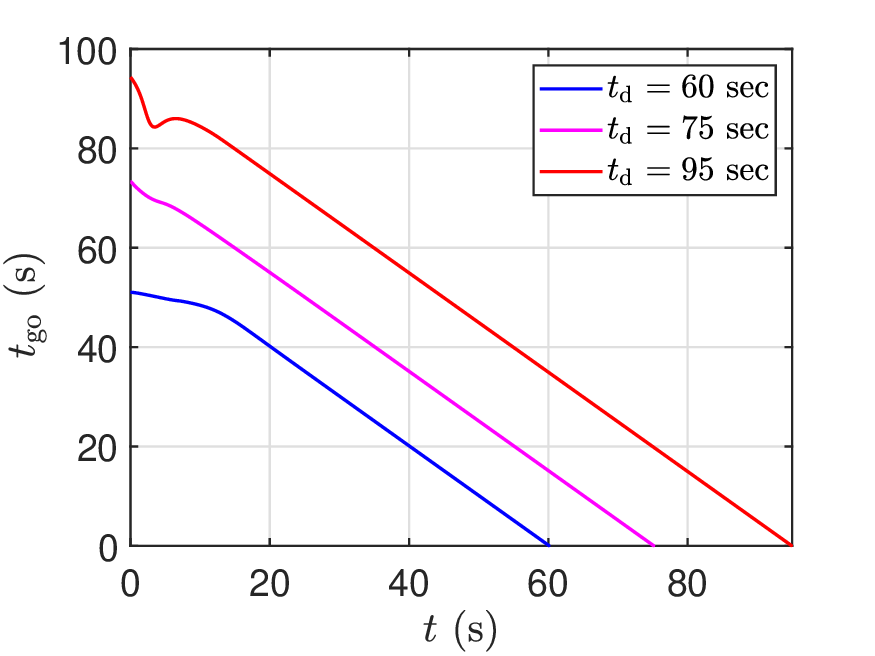}
    \caption{Time-to-go profiles.}
    \label{fig:tgo_com}
  \end{subfigure}
  \begin{subfigure}[t]{.49\linewidth}
    \centering
    \includegraphics[width=\linewidth]{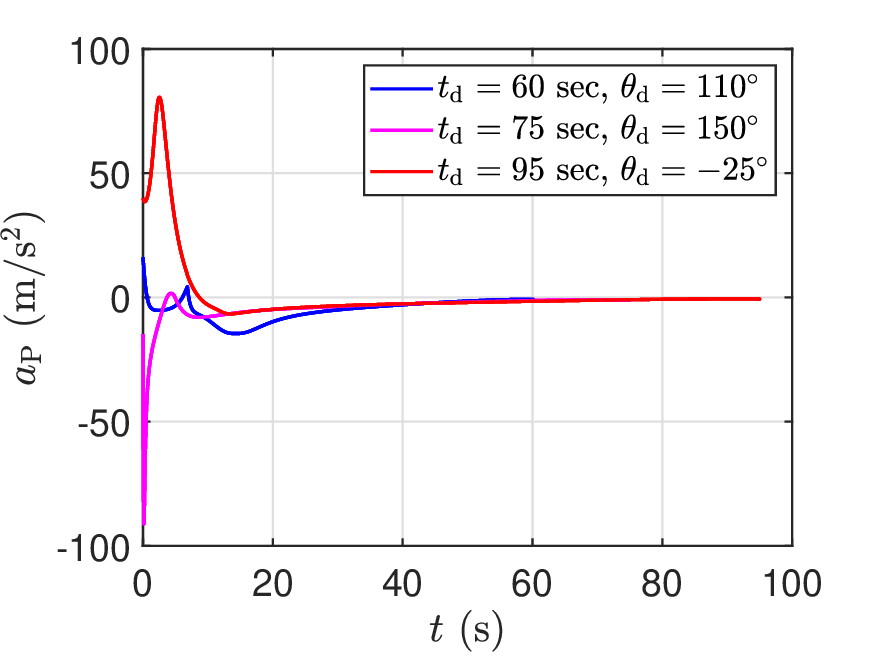}
    \caption{Lateral acceleration profiles.}
\label{fig:aP_com}
  \end{subfigure}
  \begin{subfigure}[t]{.49\linewidth}
    \centering
    \includegraphics[width=\linewidth]{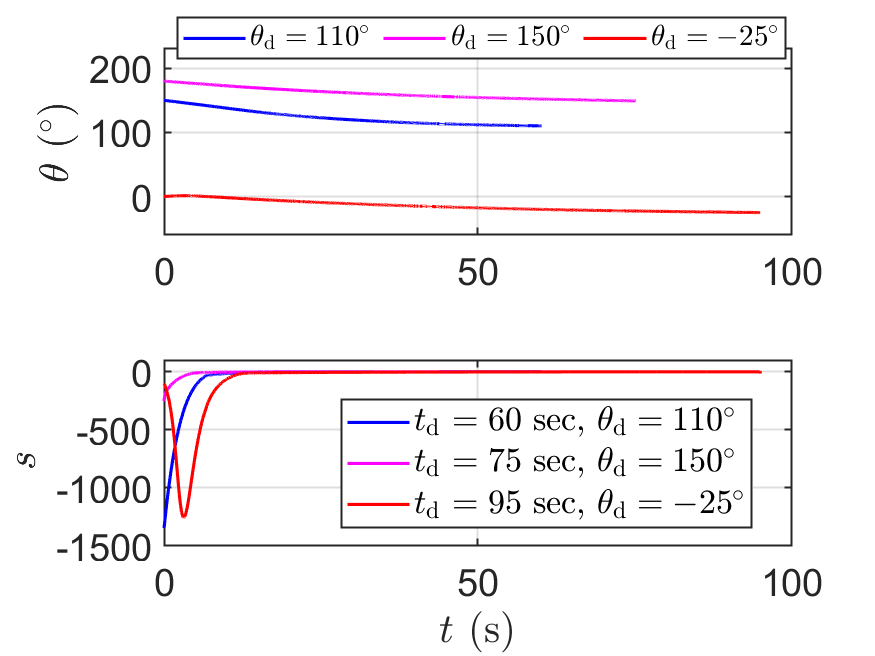}
    \caption{ LOS angles and composite sliding manifold profiles.}
    \label{fig:snet_com}
  \end{subfigure}
   \begin{subfigure}[t]{.49\linewidth}
    \centering
    \includegraphics[width=\linewidth]{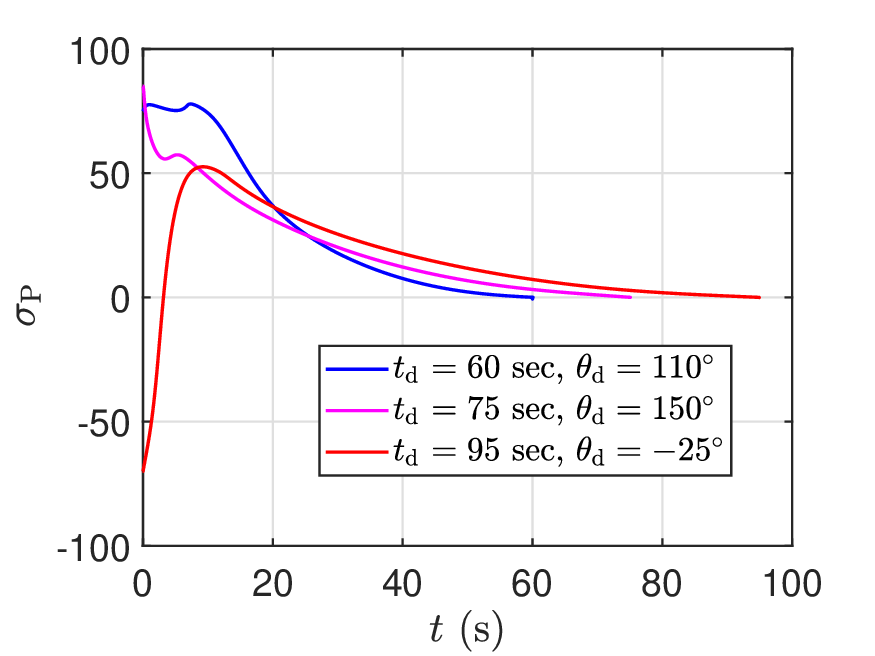}
    \caption{Lead angle profiles.}
    \label{fig:sig_com}
  \end{subfigure}
   \begin{subfigure}[t]{.49\linewidth}
    \centering
    \includegraphics[width=\linewidth]{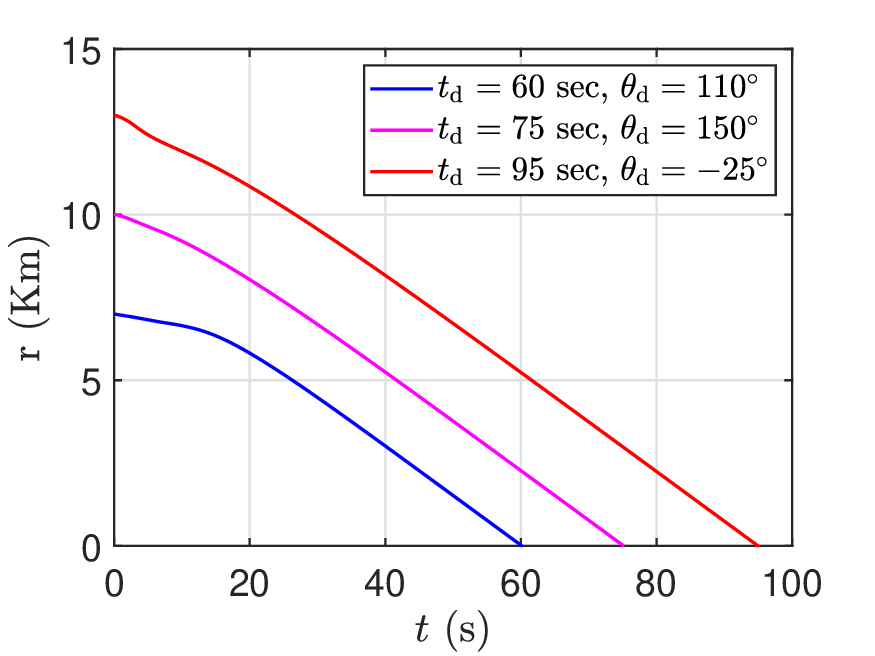}
    \caption{Range profiles.}
    \label{fig:r_com}
  \end{subfigure}
  \caption{Performance of the proposed strategy against a stationary target.}
  \label{fig:Results_com}
\end{figure}

Next, simulations are performed for the interception of a non-maneuvering target. The target's heading, impact time, and impact angle are fixed at $0 ^\circ$, $100$ s, and $135 ^\circ$, respectively. The target is moving with a speed of $65$ m/s. The initial range and the LOS angle between the interceptor and the target are $7$ km and $150^\circ$, respectively, while the initial heading of the interceptor is $45^\circ$. The design parameters are $c_1=11, c_2=9, m=185,\eta=75, \kappa=1.8, \chi_1=\chi_2=0.2$, and $\Omega=5$. The simulation results are demonstrated in \Cref{fig:Results_moving}. From the agents' trajectories plot in \Cref{fig:traj_moving}, it is evident that the moving target is intercepted successfully, indicating the effectiveness of the proposed approach. The time-to-go profile in \Cref{fig:time_to_go_moving} confirms that the interceptor achieves the target interception at the desired impact time. The control input of the interceptor is portrayed in \Cref{fig:acceleration_moving}, indicating a higher control effort requirement in the transient phase for course correction, followed by a smooth convergence to zero in the steady state. The LOS angle and the composite sliding manifold profiles are shown in \Cref{fig:errors_moving}. The LOS angle converges to its desired value monotonically to help achieve the desired impact angle. The composite sliding surface increases from a negative value and ultimately converges to zero smoothly. As shown in \Cref{fig:sig_moving}, the lead angle initially increases as the interceptor detours, followed by a smooth convergence to zero during the terminal phase. The range profile depicted in \Cref{fig:r_moving} decreases to zero as the lead angle converges to zero.

\begin{figure}[ht!]
  \centering
  \begin{subfigure}[t]{.49\linewidth}
    \centering
    \includegraphics[width=\linewidth]{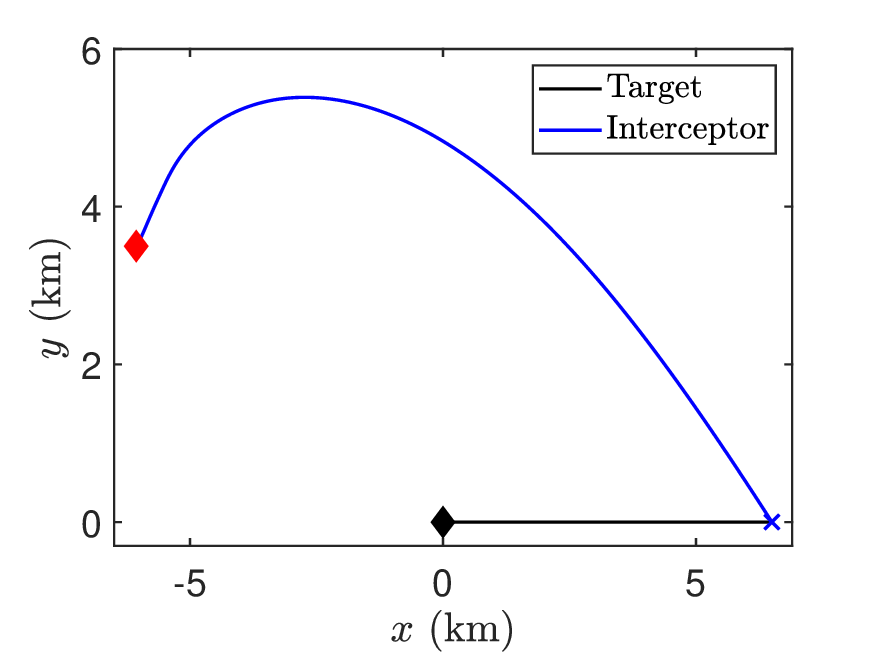}
    \caption{Agents' trajectories.}
    \label{fig:traj_moving}
  \end{subfigure}%
  \begin{subfigure}[t]{0.49\linewidth}
    \centering
    \includegraphics[width=\linewidth]{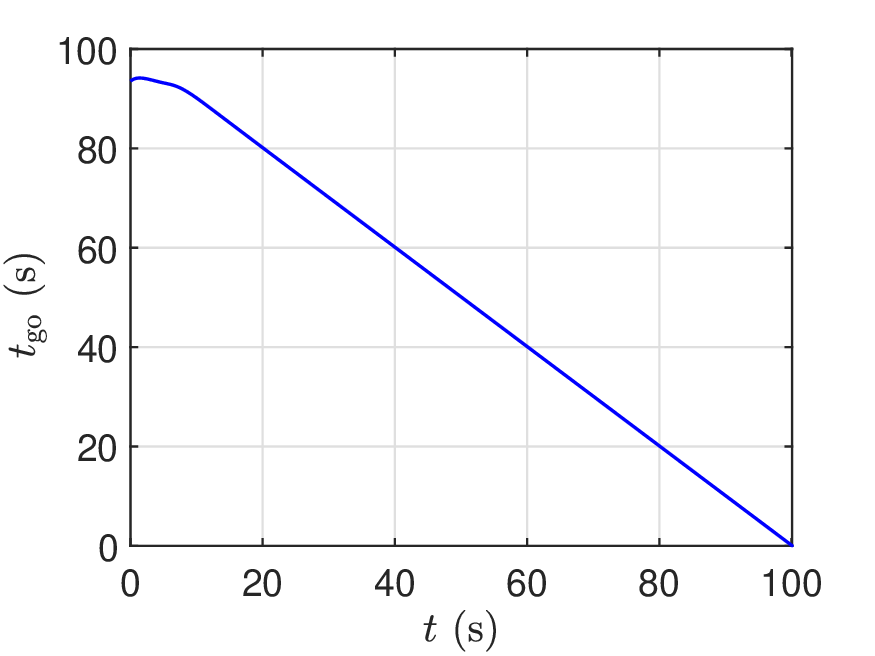}
    \caption{Time-to-go profile.}
    \label{fig:time_to_go_moving}
  \end{subfigure}
  \begin{subfigure}[t]{.49\linewidth}
    \centering
    \includegraphics[width=\linewidth]{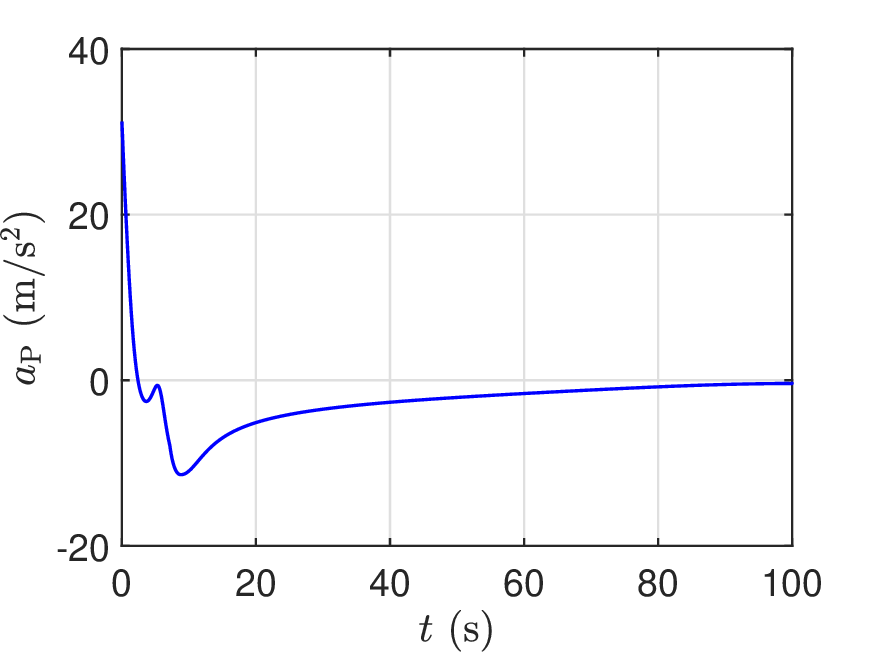}
    \caption{Lateral acceleration (steering control).}
    \label{fig:acceleration_moving}
  \end{subfigure}
  \begin{subfigure}[t]{.49\linewidth}
    \centering
    \includegraphics[width=\linewidth]{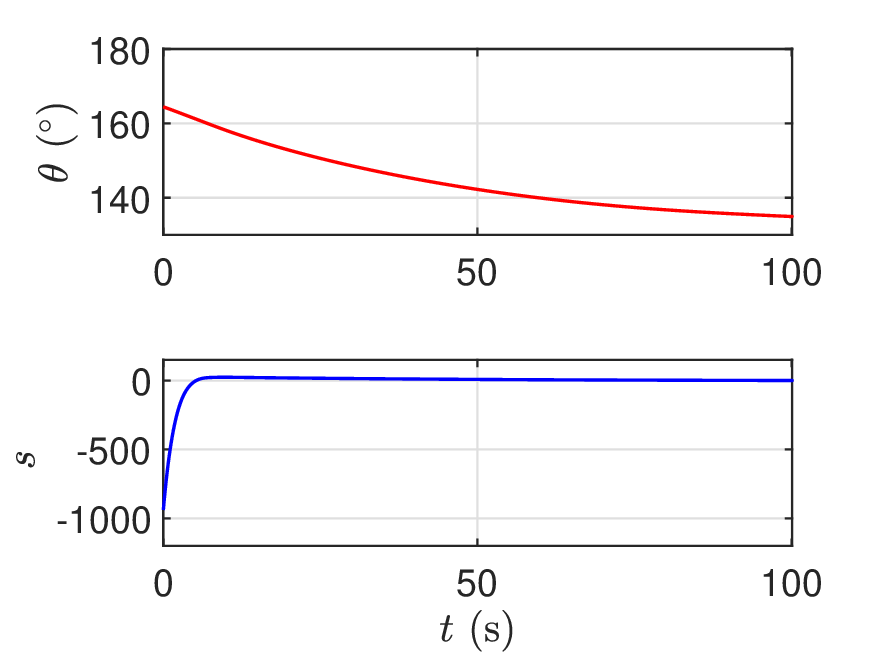}
    \caption{LOS angle and composite sliding manifold profiles.}
    \label{fig:errors_moving}
  \end{subfigure}
  \begin{subfigure}[t]{.49\linewidth}
    \centering
    \includegraphics[width=\linewidth]{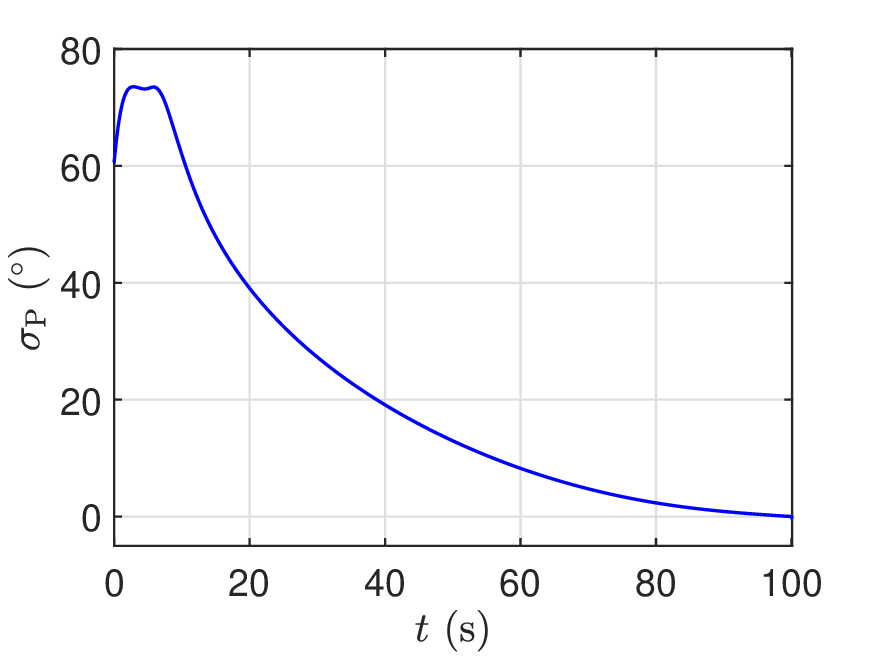}
    \caption{Lead angle profile.}
    \label{fig:sig_moving}
  \end{subfigure}
  \begin{subfigure}[t]{.49\linewidth}
    \centering
    \includegraphics[width=\linewidth]{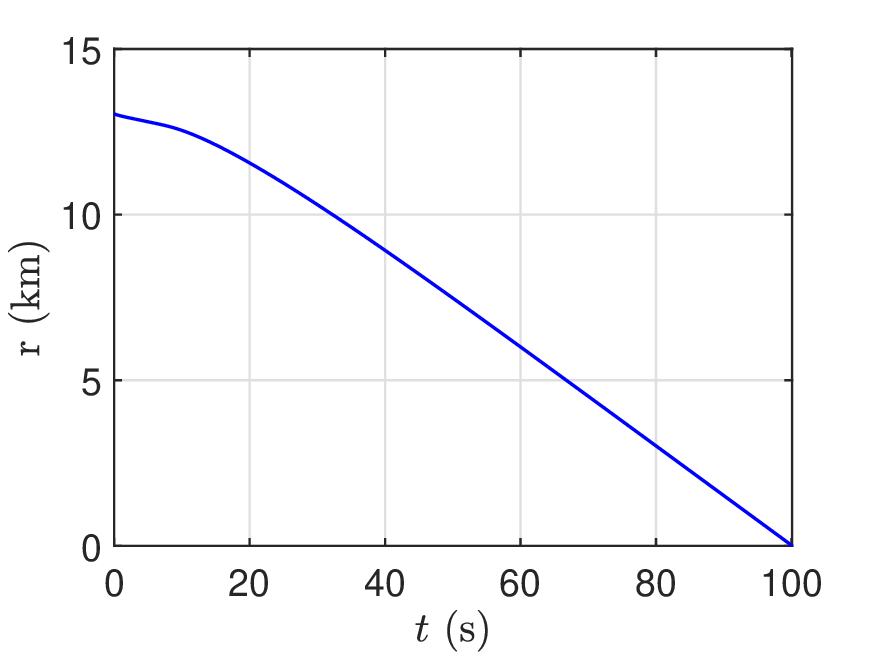}
    \caption{Range profile.}
    \label{fig:r_moving}
  \end{subfigure}
  \caption{Performance of the proposed strategy against a moving target.}
  \label{fig:Results_moving}
\end{figure}
\begin{figure}[ht!]
  \centering
  \begin{subfigure}[t]{.49\linewidth}
    \centering
    \includegraphics[width=\linewidth]{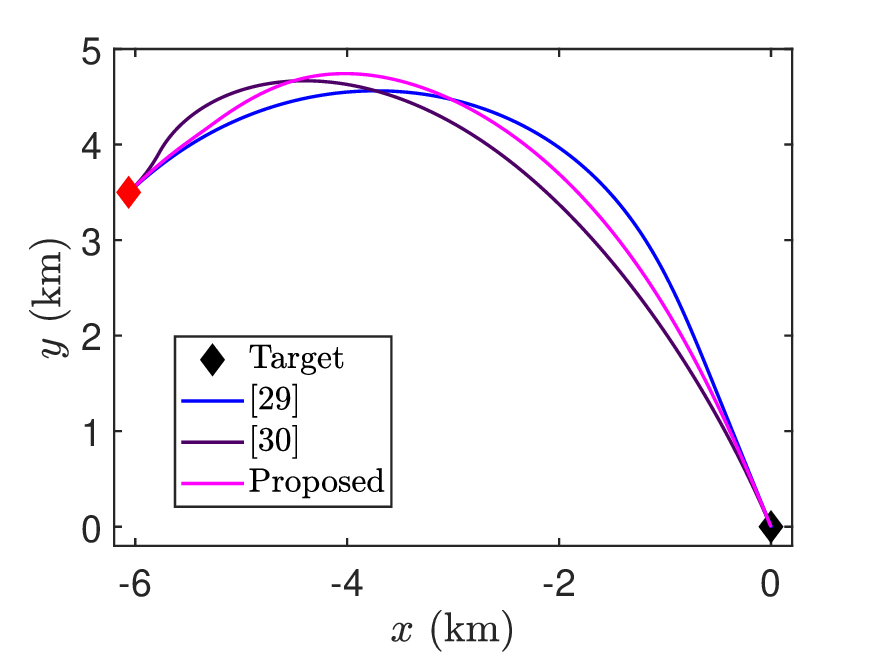}
     \caption{Trajectories.}
    \label{fig:without_auto_traj}
  \end{subfigure}
  \begin{subfigure}[t]{.49\linewidth}
    \centering
    \includegraphics[width=\linewidth]{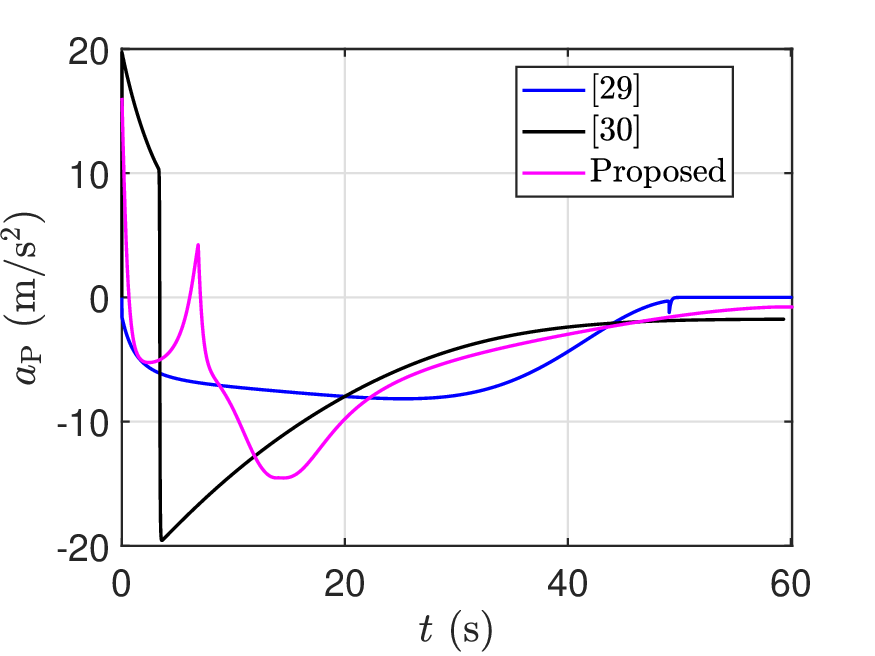}
    \caption{Lateral accelerations.}
    \label{fig:without_auto_acc}
  \end{subfigure}
  \begin{subfigure}[t]{.49\linewidth}
    \centering
    \includegraphics[width=\linewidth]{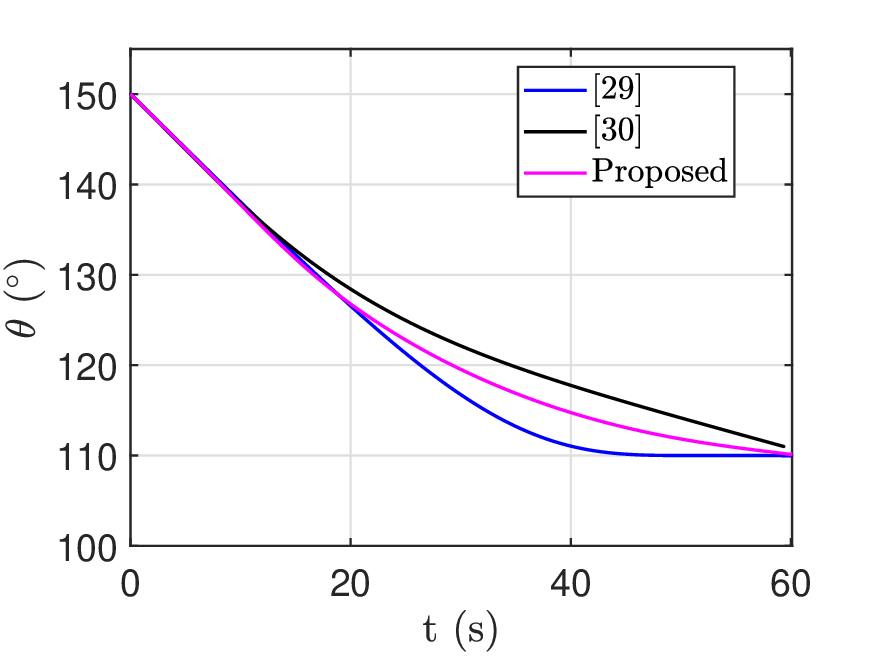}
     \caption{LOS angles.}
     \label{fig:without_auto_theta}
  \end{subfigure}
  \begin{subfigure}[t]{.49\linewidth}
    \centering
    \includegraphics[width=\linewidth]{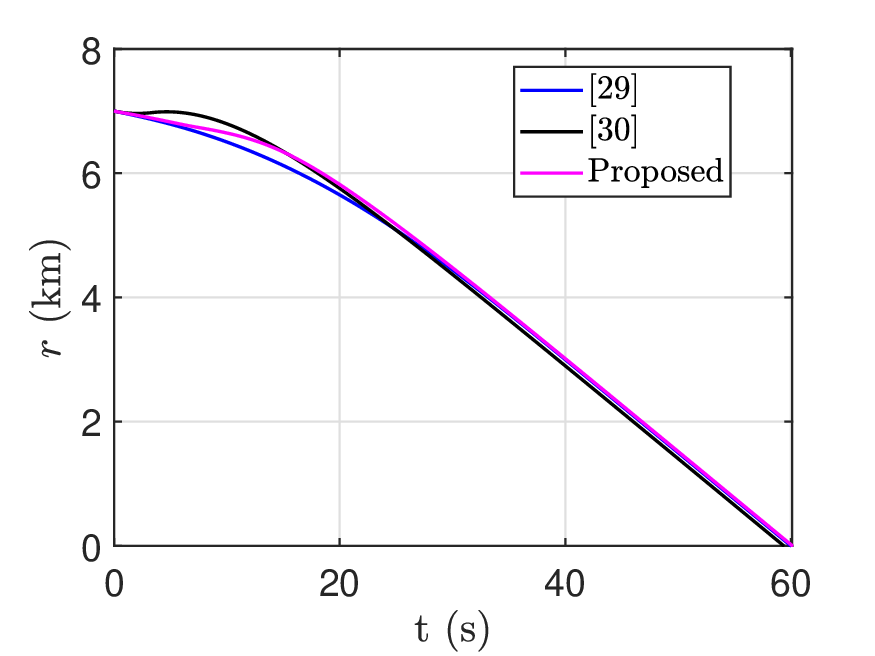}
     \caption{Range profiles.}
    \label{fig:without_auto_range}
  \end{subfigure}
  \begin{subfigure}[t]{.49\linewidth}
    \centering
    \includegraphics[width=\linewidth]{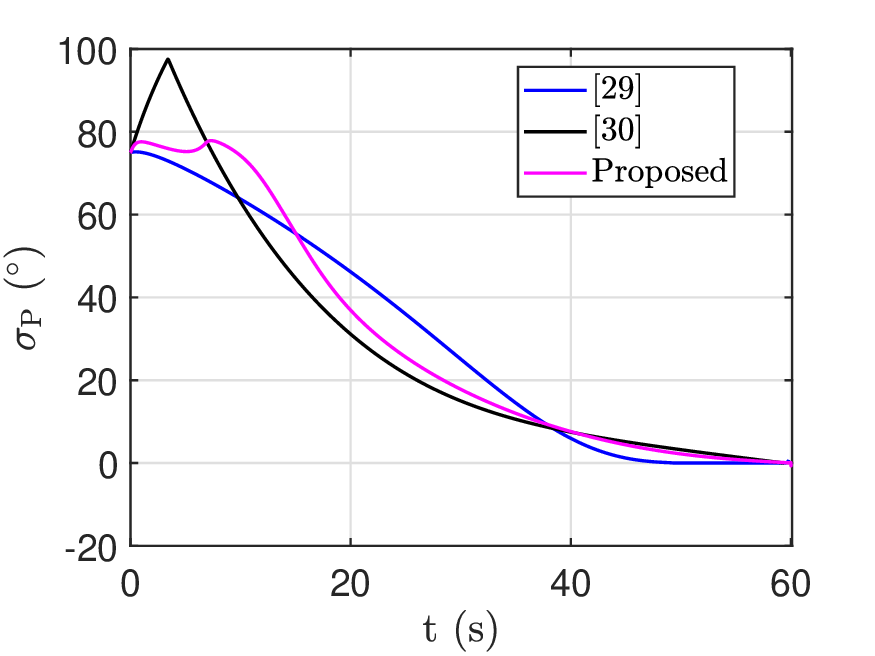}
     \caption{Lead angles' profiles.}
    \label{fig:without_auto_sig}
  \end{subfigure}
  \begin{subfigure}[t]{.49\linewidth}
    \centering
    \includegraphics[width=\linewidth]{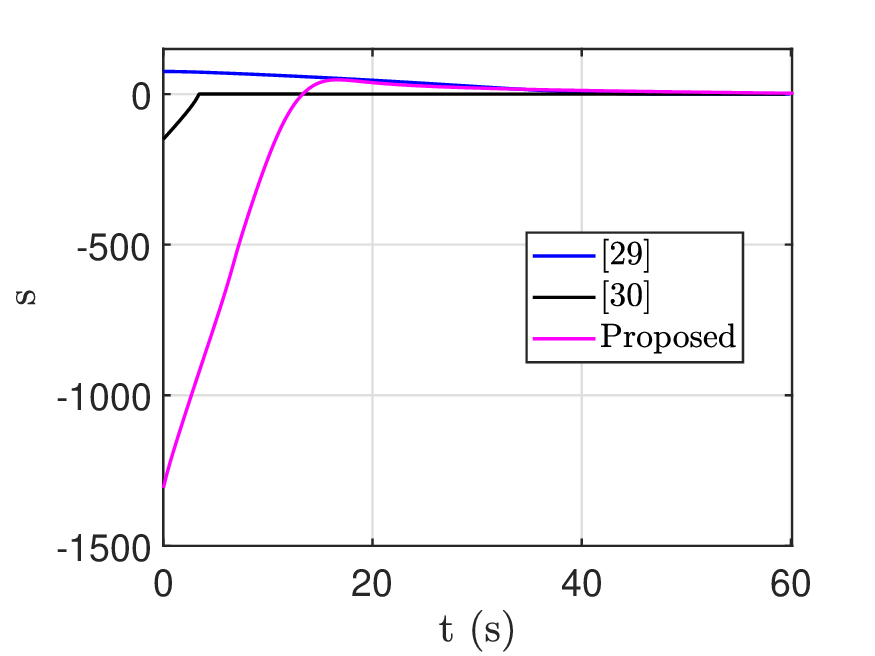}
     \caption{Sliding manifold profiles.}
    \label{fig:without_auto_s}
  \end{subfigure}
  \caption{Performance comparison of the proposed guidance strategy with the existing guidance strategies.}
   \label{fig:without_auto_comparison}
\end{figure}

The performance of the proposed strategy is also compared with the ones in \cite{Chen_Wang} and \cite{doi:10.1177/09544100211029817}, and the corresponding results are shown in \Cref{fig:without_auto_comparison}. The initial range and the LOS angle between the interceptor and the target are $7$ km and $150^\circ$, respectively, while the initial heading of the interceptor is $45^\circ$. The desired impact time and impact angle are set to $60$ s and $110^\circ$. All simulations are performed under the same conditions. It is evident from \Cref{fig:without_auto_traj} that, under all guidance strategies, the target is intercepted successfully, though the interceptor's trajectories are different. The lateral acceleration profiles are depicted in \Cref{fig:without_auto_acc}. It can be observed that the guidance strategy proposed in \cite{doi:10.1177/09544100211029817} results in a sharp drop in lateral acceleration at the time instant when the sliding mode is enforced, which is due to the time-varying weight assigned in the beginning to allow a correction to the impact time. Such correction is unnecessary in the proposed design. One may also notice that both the proposed strategy and that in \cite{Chen_Wang} have smooth lateral acceleration profiles. Furthermore, the acceleration demand is compared based on the control effort calculation using $\displaystyle\int_{t_0}^{t_\mathrm{f}} a_\mathrm{P}^2 \, dt$. Based on this integral-of-square effort, the lateral acceleration demand using the strategy in \cite{Chen_Wang} is $157 \%$ higher than the proposed strategy, whereas the control demand using \cite{doi:10.1177/09544100211029817} is $13\%$
higher. 

The LOS angle profiles are depicted in \Cref{fig:without_auto_theta}. It can be observed that the LOS angle profile converges to its desired value faster using the strategy in \cite{Chen_Wang}, while for the proposed strategy and the strategy presented in \cite{doi:10.1177/09544100211029817}, the LOS profiles converge to their desired value in the terminal phase to meet the impact angle constraint. Forcing angle correction early may require larger lateral acceleration demand, and could be too stringent. The range profiles of the interceptor are depicted in \Cref{fig:without_auto_range}, indicating that range profiles converge to zero at the prescribed impact time. The proposed approach thus demonstrates effectiveness even under large initial heading errors, successfully caters to non-maneuvering targets, and operates independently of any specific sliding surface or time-to-go formulation. The lead angle profiles are shown in \Cref{fig:without_auto_sig}. It is apparent from the profiles that the lead angle profile converges to zero faster using \cite{Chen_Wang}, as the LOS angle converges to zero, while for the proposed strategy and the strategy presented in \cite{doi:10.1177/09544100211029817}, the lead angles converge to zero in the terminal phase. The sliding manifold profiles are depicted in \Cref{fig:without_auto_s}. The sliding manifold converges to zero within around $5$ s using the strategy in \cite{doi:10.1177/09544100211029817}, while for the proposed strategy and the strategy proposed in \cite{Chen_Wang}, the sliding manifold profiles converge to zero in around $40$ s.

The performance of the proposed guidance strategy is further compared with the ones in \cite{Chen_Wang} and \cite{doi:10.1177/09544100211029817} in the presence of an autopilot modeled as a first-order lag with a time-constant of $0.1$ s. The simulation results for this case are presented in \Cref{fig:with_auto_comparison}. The initial conditions are kept the same as in the previous case of \Cref{fig:without_auto_comparison}. The trajectories of the interceptor are depicted in \Cref{fig:with_auto_traj}, where the interceptor successfully intercepts the target under all guidance strategies under the presence of autopilot lag. The lateral acceleration profiles are shown in \Cref{fig:with_auto_acc}, indicating a similar trend as in \Cref{fig:without_auto_acc}. From the control effort calculation, the lateral acceleration demand for the strategy in \cite{Chen_Wang} is now $132 \%$ higher than the proposed strategy, while the control demand for the strategy in \cite{doi:10.1177/09544100211029817}  is $32 \%$ high. The LOS angle profiles are represented in \Cref{fig:with_auto_theta}, showing a similar behavior as in the absence of autopilot lag (\Cref{fig:without_auto_theta}). The range profiles are demonstrated in \Cref{fig:with_auto_range}, where it converges to zero at the desired impact time under all guidance strategies. The lead angle profiles are shown in \Cref{fig:with_auto_sig}. Similar to the previous case (\Cref{fig:without_auto_sig}), the lead angle profile converges to zero faster using the strategy in \cite{Chen_Wang}, while for the proposed strategy and the strategy presented in \cite{doi:10.1177/09544100211029817}, the lead angles converge to zero in the terminal phase. The sliding manifold profiles are demonstrated in \Cref{fig:with_auto_s}, showing a similar trend as in the absence of autopilot lag (\Cref{fig:without_auto_s}). The performance comparison of control efforts in different guidance laws is presented in \Cref{tab:control_com}.

\begin{table}[ht!]
\centering
\renewcommand{\arraystretch}{2}
\begin{tabular}{m{1.5cm}>{\centering\arraybackslash}m{8.5cm}cc}
\toprule
\textbf{Reference} & \textbf{Guidance Law} & \textbf{\makecell[c]{Control Effort \\ without \\Autopilot (m$^2$/s)}} & \textbf{\makecell[c]{Control Effort \\ with\\ Autopilot (m$^2$/s)}} \\ \midrule
\cite{Chen_Wang} &$\begin{aligned}
    a_\mathrm{P}=&-\dfrac{v_\mathrm{P}^2\sin \sigma_\mathrm{P}}{r} -v_\mathrm{P}k\sign^{c_2/c_1}(s)\\ &+ \dfrac{4v^2_\mathrm{P} r\left(\cos \sigma_\mathrm{P}-1\right)e_\theta}{r^2e_\theta^2 + 4 \left(v_\mathrm{P}t_\mathrm{d}-v_\mathrm{P}t - r\right)^2} \\ &
    \dfrac{+ 4v^2_\mathrm{P}\left(\cos \sigma_\mathrm{P}e_\theta + \sin \sigma_\mathrm{P}\right)\left(v_\mathrm{P}t_\mathrm{d}-v_\mathrm{P}t - r\right)}{r^2e_\theta^2 + 4 \left(v_\mathrm{P}t_\mathrm{d}-v_\mathrm{P}t - r\right)^2} 
\end{aligned}$  & 8008.9 & 6112.8\\ \midrule
\cite{doi:10.1177/09544100211029817} & $\begin{aligned}
a_\mathrm{P}=&- \dfrac{k f_t + \dfrac{\left(\dot{e}_\theta -\dfrac{mc_1}{c_2}\dot{e}_\theta^{c_1/c_2 -1}\dfrac{2 \dot{r}\dot{\theta}}{r} \right)}{t_\mathrm{go}^\mathrm{d}-t_\mathrm{d}}}{\left(k - \dfrac{\left(e_\theta + m \dot{e}_\theta^{c_1/c_2} \right)}{\left(t_\mathrm{go}^\mathrm{d}-t_\mathrm{d}\right)^2}\right)g_t + h_t} \\ & -\dfrac{ - \dfrac{\left(e_\theta + m \dot{e}_\theta^{c_1/c_2} \right)}{\left(t_\mathrm{go}^\mathrm{d}-t_\mathrm{d}\right)^2}\left(f_t -1\right)}{\left(k - \dfrac{\left(e_\theta + m \dot{e}_\theta^{c_1/c_2} \right)}{\left(t_\mathrm{go}^\mathrm{d}-t_\mathrm{d}\right)^2}\right)g_t + h_t}
\end{aligned}$
  $\begin{aligned}
    \text{where }f_t &= 1 + \frac{\dot{r}}{v_\mathrm{P}}\left( 1 + \frac{\sigma_\mathrm{P}^2 + \sigma_\mathrm{P_f}^2}{15} - \frac{\sigma_\mathrm{P}\sigma_\mathrm{P_f}}{30}\right) \\& + \frac{r\dot{\theta}}{v_\mathrm{P}}\left( -\frac{2\left(\sigma_\mathrm{P} + \sigma_\mathrm{P_f}\right)}{15}  + \frac{\sigma_\mathrm{P} + \sigma_\mathrm{P_f}}{30}\right),
\end{aligned}$
$\begin{aligned}
    g_t&= \left(\frac{r}{v_\mathrm{P}^2}\left(\frac{2\sigma_\mathrm{P}}{15}-\frac{\sigma_\mathrm{P_f}}{30} \right) \right)
\end{aligned}$, ~
$\begin{aligned}
    h_t= - \frac{{mc_1}\dot{e}_\theta^{c_1/c_2 -1}\cos \sigma_\mathrm{P}}{rc_2 \left(t_\mathrm{go}^\mathrm{d}-t_\mathrm{d}\right)}
\end{aligned}$
& 3513.3 & 3479.2 \\ \midrule
Proposed & Given in 
 \eqref{eq:lateral  acceleration}
& 3108.1 & 2634.8 \\ \bottomrule
\end{tabular}
\caption{Comparison of performance of various guidance laws.}
\label{tab:control_com}
\end{table}

\begin{figure}[ht!]
  \centering
  \begin{subfigure}[t]{.49\linewidth}
    \centering
    \includegraphics[width=\linewidth]{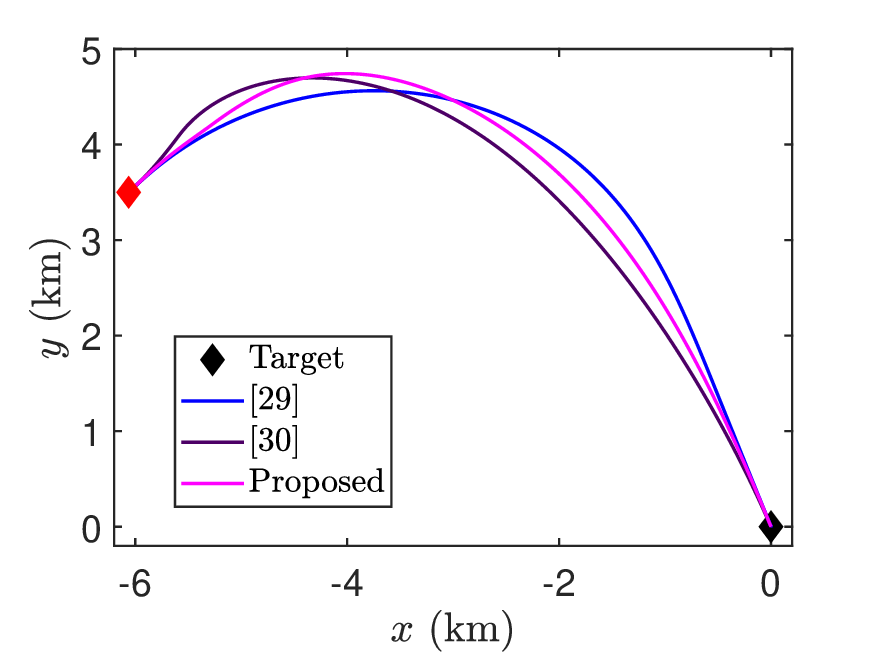}
     \caption{Trajectories.}
    \label{fig:with_auto_traj}
  \end{subfigure}
  \begin{subfigure}[t]{.49\linewidth}
    \centering
    \includegraphics[width=\linewidth]{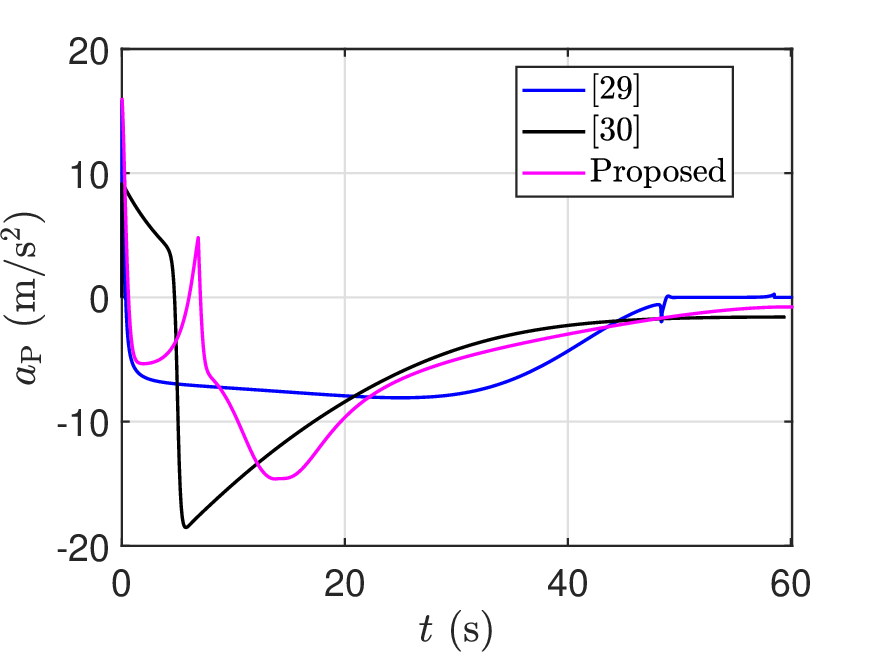}
    \caption{Lateral accelerations.}
    \label{fig:with_auto_acc}
  \end{subfigure}
  \begin{subfigure}[t]{.49\linewidth}
    \centering
    \includegraphics[width=\linewidth]{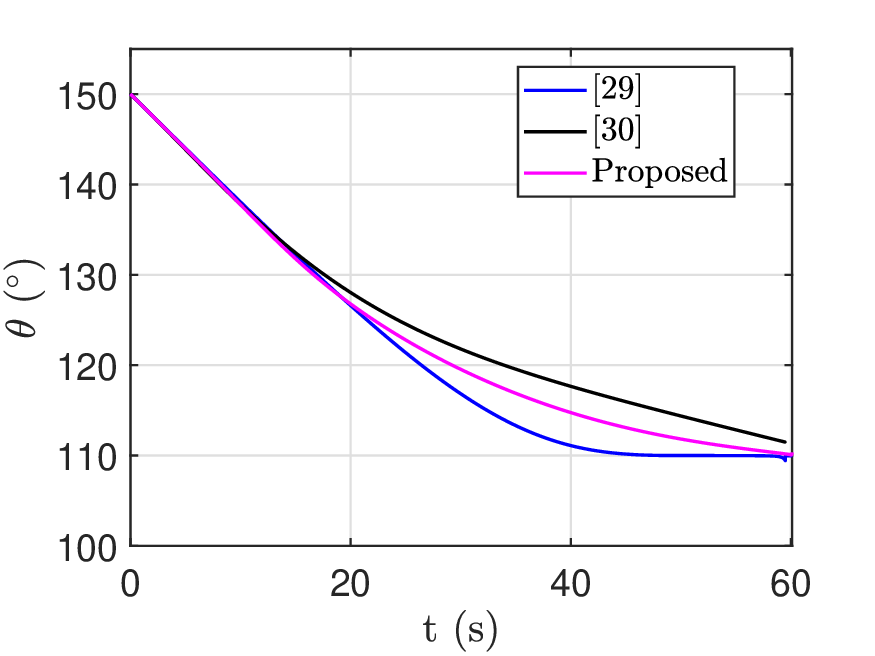}
     \caption{LOS angles.}
     \label{fig:with_auto_theta}
  \end{subfigure}
  \begin{subfigure}[t]{.49\linewidth}
    \centering
    \includegraphics[width=\linewidth]{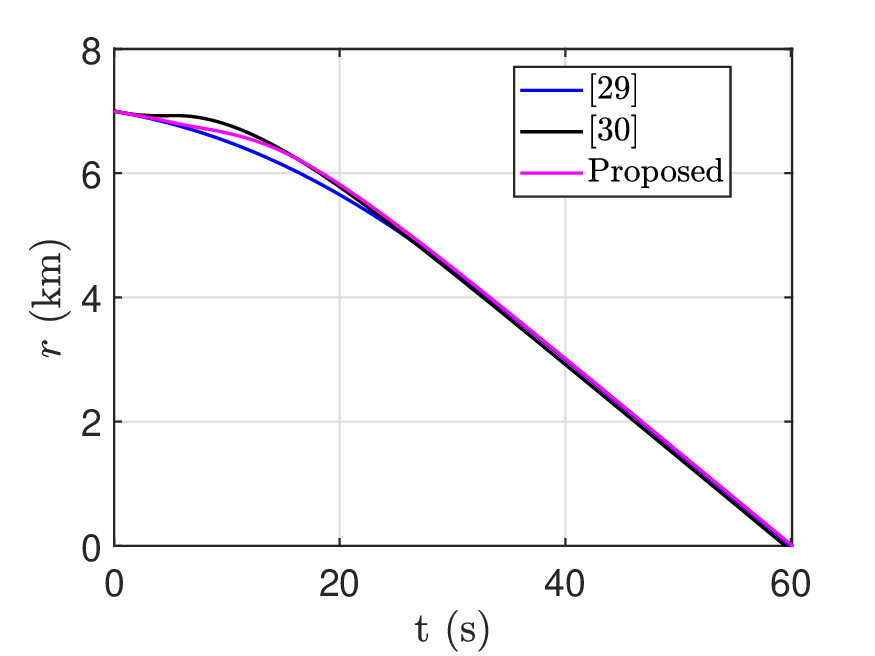}
     \caption{Range profiles.}
    \label{fig:with_auto_range}
  \end{subfigure}
  \begin{subfigure}[t]{.49\linewidth}
    \centering
    \includegraphics[width=\linewidth]{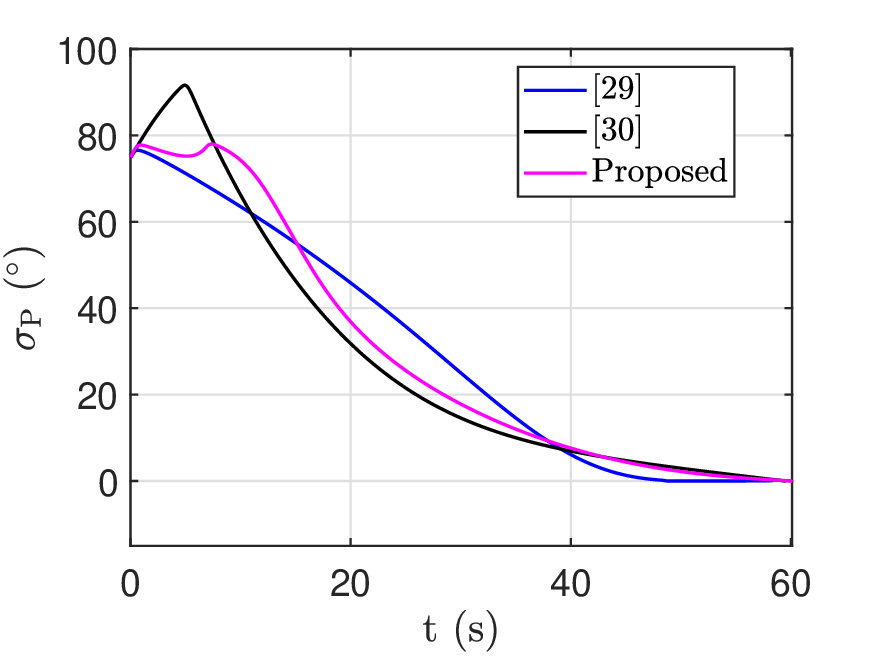}
     \caption{Lead angle profiles.}
    \label{fig:with_auto_sig}
  \end{subfigure}
  \begin{subfigure}[t]{.49\linewidth}
    \centering
    \includegraphics[width=\linewidth]{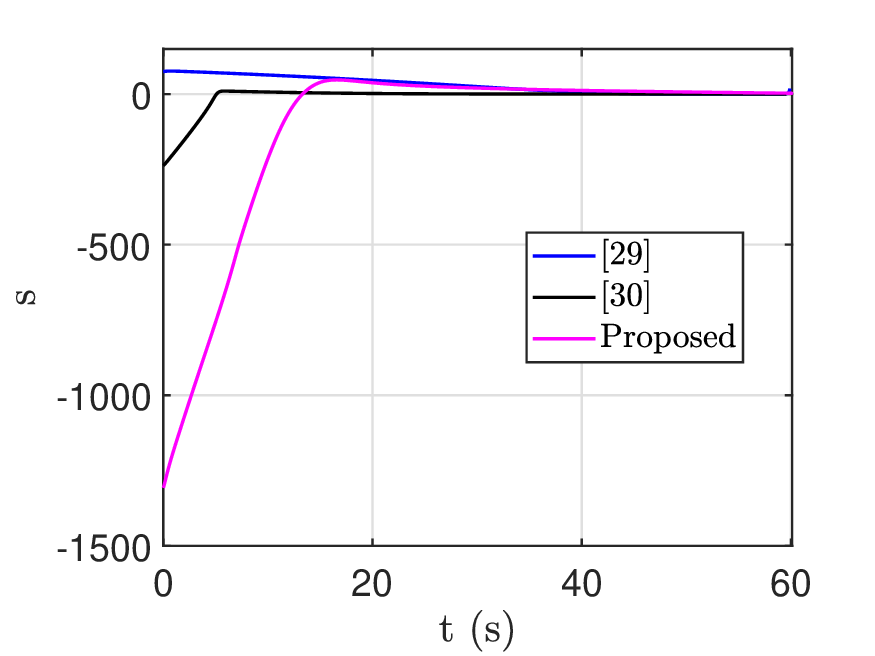}
     \caption{Sliding manifold profiles.}
    \label{fig:with_auto_s}
  \end{subfigure}
  \caption{Performance comparison of proposed strategy with existing guidance strategies under first-order autopilot lag.}
   \label{fig:with_auto_comparison}
\end{figure}

 The robustness of the proposed approach with respect to another time-to-go estimate given in \eqref{eq:tgo_2}, is validated through the simulation results next.  The initial range and LOS angle between the interceptor and the target are $6$ km and $170^\circ$, respectively, while the initial heading of the interceptor is $70^\circ$. The desired time and angle are set to $45$ s and $110^\circ$, respectively. The design parameters are $c_1=9, c_2=7, m=250,\eta=0.1, \kappa=5, \chi_1=\chi_2=0$, and $\Omega=0.75$. The simulation results for this case are shown in \Cref{fig:Results_tgo2}. The interceptor's trajectory is illustrated in \Cref{fig:tgo22}, from which one may notice that the interceptor successfully intercepts the target even if the time-to-go is different. The time-to-go profile is depicted in \Cref{fig:tgo22}, in which it is apparent that the target is intercepted at the desired impact time. Additionally, even though the initial time-to-go estimate is greater than the desired impact time, the interceptor effectively adjusts its trajectory using the strategy following the command in \eqref{eq:aP2} to ensure interception at the desired impact time. The interceptor's lateral acceleration profile is depicted in \Cref{fig:a_P_tgo2}. It can be observed that the interceptor requires higher control authority in the transient phase to achieve its course correction, followed by a smooth convergence close to zero. Furthermore, at around $40$ s, there appears a smooth increase in the lateral acceleration demand to achieve the impact angle constraint. The LOS angle and the composite sliding manifold are demonstrated in \Cref{fig:s_tgo2}. The LOS angle approaches its desired value in the terminal phase. The composite sliding surface rapidly converges to zero and remains on it. The lead angle in \Cref{fig:sig_tgo2} converges to zero at the interception instant. As depicted in \Cref{fig:r_tgo2}, the range eventually converges to zero as the lead angle decreases, which is consistent with the previous cases.
\begin{figure}[ht!]
  \centering

  \begin{subfigure}[t]{.49\linewidth}
    \centering
    \includegraphics[width=\linewidth]{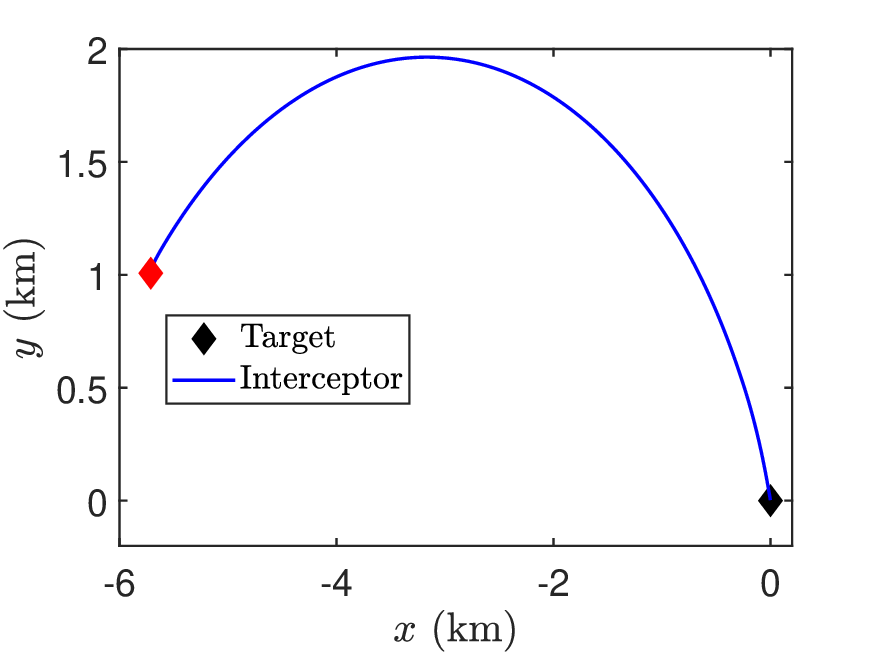}
    \caption{Interceptor's trajectory.}
    \label{fig:traj_tgo2}
  \end{subfigure}\hfill 
  \begin{subfigure}[t]{0.49\linewidth}
    \centering
    \includegraphics[width=\linewidth]{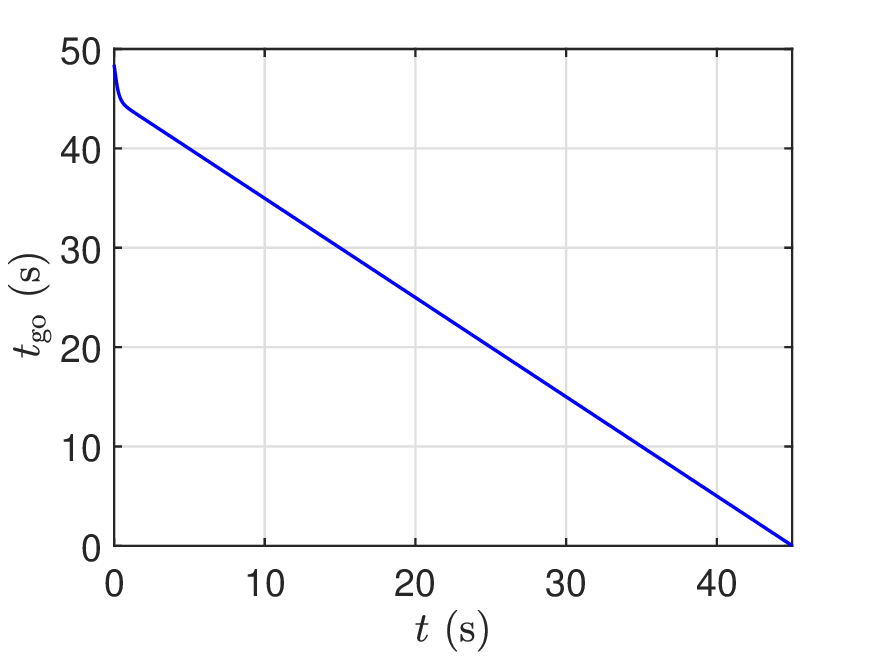}
    \caption{Time-to-go profile.}
    \label{fig:tgo22}
  \end{subfigure}

  \begin{subfigure}[t]{.49\linewidth}
    \centering
    \includegraphics[width=\linewidth]{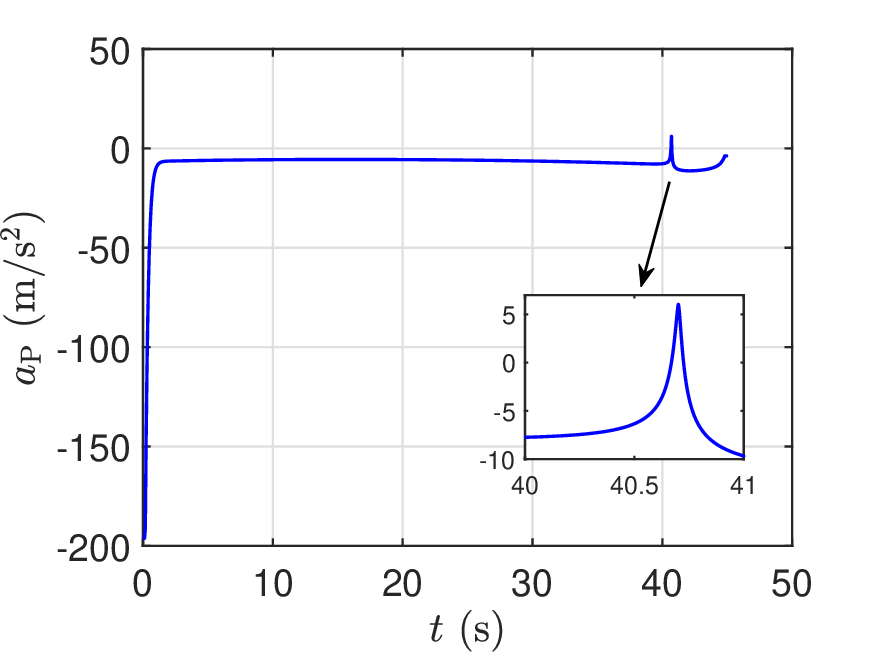}
    \caption{Lateral acceleration (steering control).}
    \label{fig:a_P_tgo2}
  \end{subfigure}\hfill
  \begin{subfigure}[t]{.49\linewidth}
    \centering
    \includegraphics[width=\linewidth]{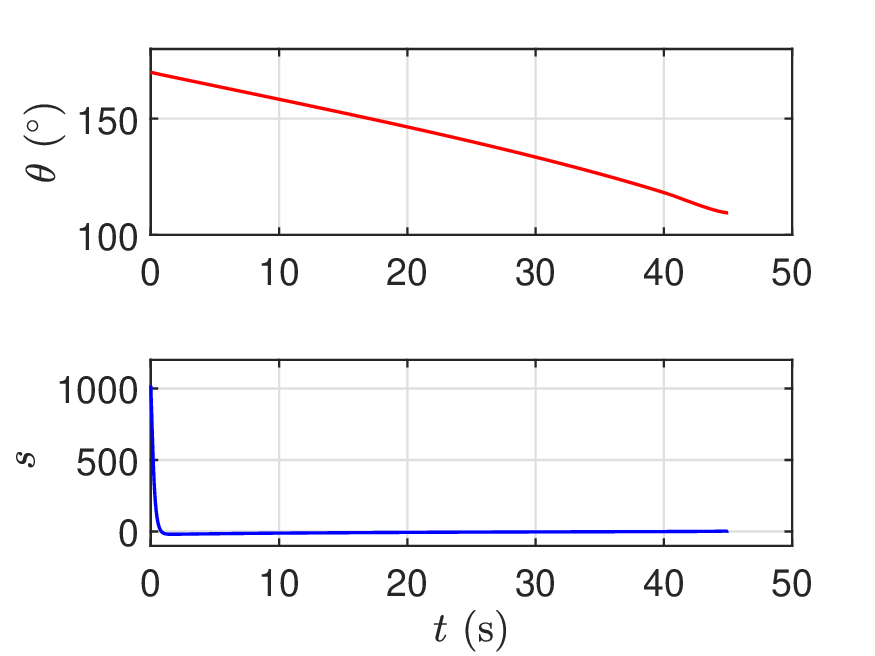}
    \caption{LOS angle and composite sliding manifold profiles.}
    \label{fig:s_tgo2}
  \end{subfigure}
  \begin{subfigure}[t]{.49\linewidth}
    \centering
    \includegraphics[width=\linewidth]{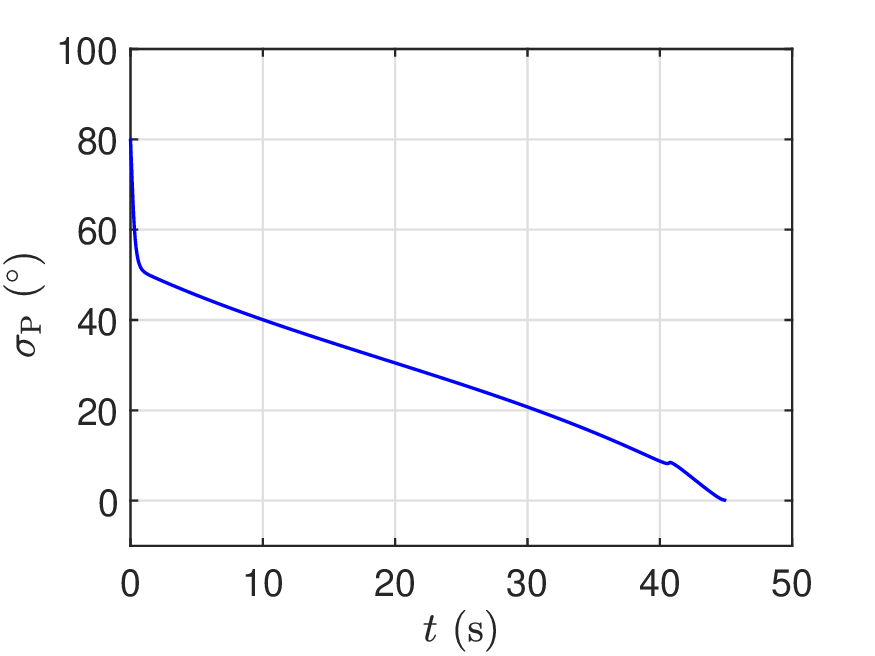}
    \caption{Lead angle profile.}
    \label{fig:sig_tgo2}
  \end{subfigure}
  \begin{subfigure}[t]{.49\linewidth}
    \centering
    \includegraphics[width=\linewidth]{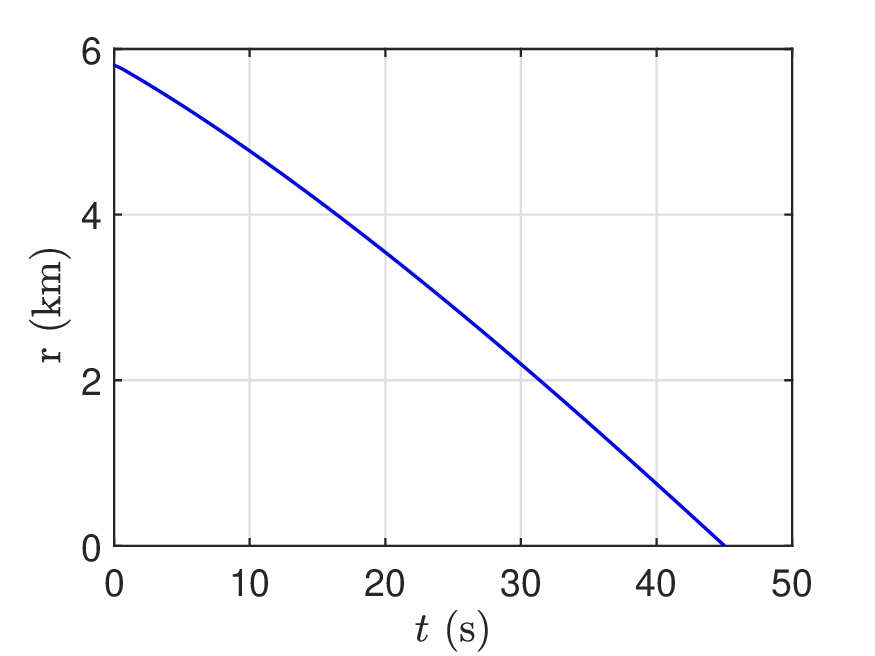}
    \caption{Range.}
    \label{fig:r_tgo2}
  \end{subfigure}
  \caption{Performance of the proposed strategy for the time-to-go \eqref{eq:tgo_2} for $t_\mathrm{f}=45$ s and $\theta_\mathrm{d}= 110^\circ$.}
  \label{fig:Results_tgo2}
\end{figure}

The variations of the adaptive gain ($\lambda$) for different initial and terminal conditions are depicted in \Cref{fig:lambda}. It is apparent from the plot that the gain varies during the transient phase as the interceptor adjusts its course to drive the impact time error to zero, followed by convergence to a constant value in the steady state. Additionally, it is evident that the gain does not cross zero, indicating that the sub-sliding surface associated with the impact angle error, $s_\theta$, converges to zero when the sub-sliding surface associated with the impact time error, $s_\mathrm{t}$, approaches zero, which is consistent with \Cref{rem:s_rel}.
\begin{figure}[h!]
    \centering
    \includegraphics[width=0.5\linewidth]{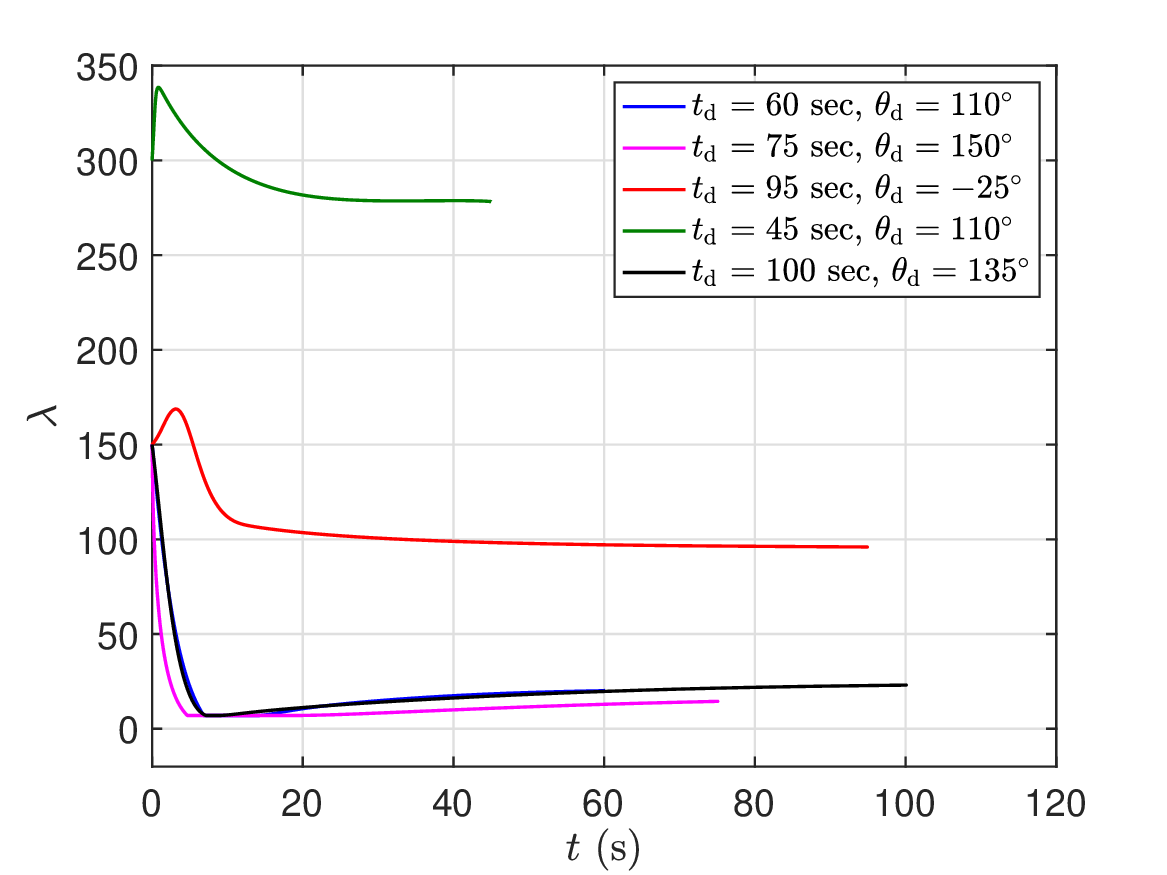}
    \caption{Variation of adaptive gain ($\lambda$).}
    \label{fig:lambda}
\end{figure}

\section{Conclusions}
In this work, we designed an adaptive guidance strategy to control both impact time and impact angle simultaneously using the interceptor's lateral acceleration as the sole control input. The proposed approach was built upon a hierarchical sliding mode framework comprising two layers-- the first layer consists of two sub-sliding surfaces associated with impact time and impact angle error variables, while the second layer integrates these sub-sliding surfaces into a composite sliding surface.  An adaptive gain was assigned to the sub-sliding surface associated with impact time. The simulation results demonstrate the effectiveness of the proposed framework for various initial engagement geometries and different time-to-go estimates. Furthermore, the performance of the proposed approach was compared with that of the existing strategies in the absence and presence of system lag. It was found that the proposed strategy is the most cost efficient in terms of lateral acceleration demand and robust to time-to-go estimates used for guidance design. Moreover, no specialized sliding manifold needs to be constructed for specific target motion, and hence, the proposed strategy was also shown to be effective against a moving target.

\bibliography{reference}
 \end{document}